\documentclass[fleqn, a4paper]{article}
\usepackage{a4wide}
\usepackage{times}
\usepackage{url}
\usepackage{tikz}
\usepackage{amsmath,amssymb}
\usepackage{amsfonts}
\usepackage{xspace}
\usepackage{subfig}
\usepackage{lscape}
\usepackage[labelfont=bf]{caption}

\newcommand{\qed}{\ensuremath{\hfill\square}}

\newcounter{theoremcnt}[section]
\renewcommand{\thetheoremcnt}{\thesection.\arabic{theoremcnt}}

\newenvironment{definition}[1]%
{\begin{trivlist}\refstepcounter{theoremcnt}
\item[]\textbf{Definition~\thetheoremcnt~(#1).}}{\end{trivlist}}

\newenvironment{lemma}%
{\begin{trivlist}\refstepcounter{theoremcnt}
\item[]\bf Lemma~\thetheoremcnt.~\rm}{\end{trivlist}}
\newenvironment{proof}%
{\begin{trivlist}\item[]{\bf Proof.}}{\qed\end{trivlist}}

\newenvironment{theorem}%
{\begin{trivlist}\refstepcounter{theoremcnt}
\item[]\bf Theorem~\thetheoremcnt.~\rm}{\end{trivlist}}

\newenvironment{lemma-repeat}[1]%
{\begin{trivlist}
\item[]\bf Lemma~#1.~\rm}{\end{trivlist}}

\newenvironment{compactenum}
{\begin{enumerate}}
{\end{enumerate}}

{\begin{trivlist}
\item[]{\bf Acknowledgements.} }{\end{trivlist}}

\newcommand{\bis}{\mathrel{\raisebox{-.2ex}[1.5ex][0ex]
                 {$\stackrel{\raisebox{-0.75ex}[1ex][0ex]
                               {\normalsize $\leftrightarrow$}}
                          {\mbox{--\hspace{-0.40ex}--}}
                  $}}}
\renewcommand{\bis}{\mathrel{\,%
  \raisebox{.3ex}{$\underline{\makebox[.7em]{$\leftrightarrow$}}$}\,}}

\newcommand{\pijl}[2]{\mathord{\: \stackrel{#1}{\rightarrow_{#2}} \:}}
\newcommand{\npijl}[1]{\mathord{\: \stackrel{#1}{\not\rightarrow} \:}}
\renewcommand{\pijl}[2]{\mathord{\: \xrightarrow{#1}_{#2} \:}}

\newcommand{\naarpijl}[1]{\mathord{\: \stackrel{#1}{\leftarrow} \:}}

\newcommand{\Nat}{\mathbb{N}}

\newcommand{\splitpi}{\textit{split}\xspace}
\newcommand{\cosplitpi}{\textit{cosplit}\xspace}
\newcommand{\setB}{{\ensuremath{\boldsymbol{B}}}\xspace}
\newcommand{\setBp}{{\ensuremath{\boldsymbol{B'}}}\xspace}
\newcommand{\setC}{{\ensuremath{\boldsymbol{C}}}\xspace}

\newcommand{\B}{{\ensuremath{B}}\xspace}

\newcommand{\R}[2]{{\ensuremath{#1\ R\ #2}}\xspace}

\newcommand{\bottomstates}{\textit{btm-sts}\xspace}
\newcommand{\nonbottomstates}{\textit{non-btm-sts}\xspace}
\newcommand{\newbottomstates}{\textit{new-btm-sts}\xspace}

\newcommand{\markedbottomstates}{\textit{mrkd-btm-sts}\xspace}
\newcommand{\markednonbottomstates}{\textit{mrkd-non-btm-sts}\xspace}
\newcommand{\constellationref}{\textit{constln-ref}\xspace} 
\newcommand{\coconstellationref}{\textit{coconstln-ref}\xspace} 
\newcommand{\inconstellationref}{\textit{inconstln-ref}\xspace} 
\newcommand{\reff}{\textit{ref}\xspace} 
\newcommand{\transitionlist}{\textit{trans-list}\xspace} 
\newcommand{\toconstellationcount}{\textit{to-constln-cnt}\xspace}
\newcommand{\toconstellations}{\textit{to-constlns}\xspace} 
\newcommand{\toconstellationref}{\textit{to-constln-ref}\xspace} 

\newcommand{\block}{\textit{block}\xspace}
\newcommand{\blocks}{\textit{blocks}\xspace}
\newcommand{\inertcounter}{\textit{inert-cnt}\xspace}

\newcommand{\constellationcounter}{\textit{constln-cnt}\xspace}
\newcommand{\coconstellationcounter}{\textit{coconstln-cnt}\xspace}
\newcommand{\splittableblocks}{\textit{splittable-blks}\xspace}
\newcommand{\constellation}{\textit{constln}\xspace}

\newcommand{\nontrivialconstellations}{\textit{non-trivial-constlns}\xspace}
\newcommand{\trivialconstellations}{\textit{trivial-constlns}\xspace}

\newcommand{\size}{{\it size}\xspace}

\begin{document}

\pagestyle{headings}

\date{}
\title{\sf An $O(m\log n)$ Algorithm for \\  Stuttering Equivalence and Branching Bisimulation}
%\titlerunning{An $O(m\log n)$ Algorithm for Stuttering Equivalence}

\author{\sf Jan Friso Groote and Anton Wijs\\
\footnotesize{ Department of Mathematics and Computer Science,}
\footnotesize{Eindhoven University of Technology}\\
\footnotesize{P.O.~Box 513, 5600 MB Eindhoven, The Netherlands}\\
\footnotesize{\tt \{J.F.Groote,A.J.Wijs\}@tue.nl}}
%\affil[2]{Departement of Mathematics and Computer Science\\
%Eindhoven University of Technology\\
%P.O.~Box 513, 5600 MB Eindhoven, The Netherlands}

%\authorrunning{J.\,F.~Groote and A.\,J.~Wijs} 
%\Copyright{J.F. Groote and A.J. Wijs}

%\subjclass{D.2.4 Software/program verification}
%\keywords{Branching bisimulation; Algorithm}
%\institute{Departement of Mathematics and Computer Science\\
%Eindhoven University of Technology\\
%P.O.~Box 513, 5600 MB Eindhoven, The Netherlands\\
%\email{\{J.F.Groote, A.J.Wijs\}@tue.nl}
%}
%\date{}

\maketitle

\begin{abstract}
\noindent
We provide a new algorithm to determine stuttering equivalence with time complexity 
$O(m \log n)$, where
$n$ is the number of states and $m$ is the number of transitions of a Kripke structure.
This algorithm can also be used to determine branching bisimulation in $O(m(\log |\mathit{Act}|+\log n))$ time
where $\mathit{Act}$ is the set of actions in a labelled transition system.

Theoretically, our algorithm substantially improves upon existing algorithms which all have time complexity 
$O(m n)$ at best \cite{BO05,BP09,GV90}. Moreover, it has better or equal space complexity.
Practical results confirm these findings showing that our algorithm can outperform existing
algorithms with orders of magnitude, especially when the sizes of the Kripke structures are large. 

The importance of our algorithm stretches far beyond
stuttering equivalence and branching bisimulation. The known $O(m n)$ algorithms were 
already far more efficient (both in space and
time) than most other algorithms to determine behavioural equivalences (including weak bisimulation) and
therefore it was often used as an essential preprocessing step. This new algorithm makes this use of
stuttering equivalence and branching bisimulation even more attractive.
\end{abstract}
%-------------------------------------------------------------------------------------------------------
\section{Introduction}
Stuttering equivalence \cite{BCG88} and branching bisimulation \cite{GW96} were proposed as alternatives to 
Milner's weak bisimulation \cite{Mi80}. They are very close to weak bisimulation, 
with as essential difference that all states in the mimicking sequence
$\tau^*a\,\tau^*$ must be related to either the state before or directly after the $a$ from the first system.
This means that branching bisimulation and stuttering equivalence are slightly stronger notions than weak bisimulation. 
%OUTCOMMMENTED TO SHORTEN TACAS SUBMISSION.
%For practical use
%there is no difference in the use of weak or branching bisimulation.
% END OUTCOMMENT

In \cite{GV90} an $O(m n)$ time algorithm was proposed for stuttering equivalence and branching bisimulation,
where $m$ is the number of transitions and $n$ is the number of states in either a Kripke structure (for stuttering equivalence) or a labelled transition 
system (for branching bisimulation). We refer to this algorithm as
 GV. It is based upon the $O(m n)$ algorithm for bisimulation equivalence in 
\cite{KS90}. Both algorithms require $O(m{+}n)$ space.
They calculate for each state whether it is bisimilar to another state.

The basic idea of the algorithms of \cite{GV90,KS90} is to partition the set of states into blocks. States that are bisimilar
always reside in the same block. Whenever there are some states in a block $B'$ from which a transition is possible to some
block $B$ and there are other states in $B'$ from which such a step is not possible, $B'$ is split accordingly. Whenever
no splitting is possible anymore, the partition is called stable, and two states are in the same block iff they are bisimilar.
%%OUTCOMMENTED FOR TACAS VERSION
%An important notion for GV is the so-called bottom state, a state that has no outgoing transition
%to a state in the same block. We refer to such transitions as inert transitions. 
%If there are no cycles of states with the same label
%in the Kripke structure, which can be eliminated in linear time, the concept of bottom states makes
%it possible to find whether a block can be split
%without having to carry out a costly transitive closure operation for the inert transitions (which causes
%the calculation of weak bisimulation to be less efficient).
%%END OUTCOMMENT

There have been some attempts to come up with improvements of GV.
The authors of \cite{BO05} observed that GV only splits a block in two parts
at a time. They proposed to split a block in as many parts as possible, reducing moving states and transitions
to new blocks. 
%OUTCOMMMENTED TO SHORTEN TACAS SUBMISSION.
%They do this by computing so-called signatures of states capturing which blocks can be reached via outgoing
%transitions, and partitioning the states using a hash table based on those signatures.
% END
Their worst case time and space complexities are worse than that of GV, especially the space
complexity $O(m n)$, but in practice
this algorithm can outperform GV. In \cite{BP09}, the space complexity is brought back to $O(m{+}n)$.
A technique to be performed on Graphics Processing Units based on both GV and~\cite{BO05,BP09} is proposed in~\cite{W15}.
This improves the required runtime considerably by employing parallelism, 
but it does not imply any improvement to the single-threaded algorithm.

In \cite{PT87} an $O(m \log n)$ algorithm is proposed for strong bisimulation as an improvement upon the 
algorithm of \cite{KS90}. The core idea for this improvement is described as ``process the smaller half''~\cite{H74}. 
Whenever a block is split in two parts the amount of work must be contributed to the size of the smallest
resulting block. In such a case a state is only involved in the process of splitting if it resides in a block at most half the size of the block it was previously in when involved in splitting. This means that a state can never be involved in 
more than $\log_2 n$ splittings. As the time used in each state is proportional to the number of incoming or outgoing transitions
in that state, the total required time is $O(m \log n)$.

In this paper we propose the first algorithm for stuttering equivalence and branching bisimulation in which the
``process the smaller half''-technique is used. By doing so, we can finally confirm the conjecture in~\cite{GV90}
that such an improvement of GV is conceivable.
Moreover, we achieve an even lower complexity, namely $O(m \log n)$, than conjectured in~\cite{GV90} by applying the technique
twice, the second time for handling the presence of inert transitions. 
%OUTCOMMENTED IN THE PAST.
First we establish whether a block can be
split by combining the approach
regarding bottom states from GV with the detection approach in~\cite{PT87}.
Subsequently, we use the ``process the smaller half''-technique again to split a block by only traversing transitions in a 
time proportional to the size of the smallest subblock. As it is not known which of the two subblocks is smallest, 
the transitions of the two subblocks are processed alternatingly,
such that the total processing time can be contributed to the smallest block. For checking behavioural equivalences, applying
such a technique is entirely new. We are only aware of a similar approach for an algorithm in which the smallest bottom strongly
connected component of a graph needs to be found~\cite{CH11}.
% END

The algorithm that we propose is complex. Although the basic sketch of the algorithm is relatively straightforward, it
heavily relies on auxiliary data structures. For instance, for each transition it is recalled how many other transitions
there are from the block where the transition starts to the constellation in which the transition ends. Maintaining such auxiliary
data structures is a tedious process. Therefore, we do not only prove the major steps of the algorithm correct,
we also provide a very detailed description of the algorithm, and we ran the implemented algorithm on many thousands of 
randomly generated 
test cases, comparing the reductions with the outcomes of existing algorithms. This not only convinced us that the algorithm 
is correct and correctly implemented, it also allows others to easily reimplement the algorithm.

From a theoretical viewpoint our algorithm outperforms its predecessors substantially. But a fair question is whether 
this also translates into practice. The theoretical bounds are complexity upperbounds, and depending on the transition system, 
the classical algorithms can be much faster than the upperbound suggests. Furthermore, the increased bookkeeping in the
new algorithm may be such a burden that all gains are lost. For this reason we compared the practical performance of
our algorithm with that of the predecessors, and we found that for practical examples
our algorithm can always match the best running times, but especially when the Kripke
structures and transition systems get large, our algorithm tends to outperform existing algorithms with orders of magnitude.

Compared to checking other equivalences the existing algorithms for branching bisimulation/stuttering equivalence were already 
known to be 
practically very efficient. This is the reason that they are being used in multiple explicit-state model checkers, such as
\textsc{Cadp}~\cite{CADP}, the \textsc{mCRL2} toolset~\cite{GM14} and \textsc{TVT}~\cite{Tamperetool}. In particular they
are being used as preprocessing steps for other equivalences (weak bisimulation, trace based equivalences) that are much
harder to compute. For weak bisimulation recently 
a $O(mn)$ algorithm has been devised \cite{L09,RT08}, but until that time an expensive transitive
closure operation of at best $O(n^{2.373})$ was required. Using our algorithm as a preprocessing step,
the computation time of all other behavioural `weak' equivalences can be made faster.

\section{Preliminaries}
We introduce Kripke structures and (divergence-blind) stuttering equivalence. Labelled transitions systems and 
branching bisimulation will
be addressed in section \ref{sec:bbisim}.
\begin{definition}{Kripke structure}
A \textit{Kripke structure} is a four tuple $K=( S,\mathit{AP},\pijl{}{}, L)$, where 
\begin{compactenum}
\item
$S$ is a finite set of states.
\item
$\mathit{AP}$ is a finite set of atomic propositions.
\item
$\pijl{}{}\subseteq S\times S$ is a total transition relation, i.e., for
each $s\in S$ there is an $s'\in S$ s.t.\ $s\pijl{}{}s'$.
\item 
$L:S\rightarrow 2^{\mathit{AP}}$ is a state labelling. 
\end{compactenum}
\end{definition}
We use $n{=}|S|$ for the number of states and $m{=}|{\pijl{}{}}|$ for the number of transitions.
For a set of states $B{\subseteq} S$, we write $s\pijl{}{B}s'$ for $s\pijl{}{}s'$ and $s'\in B$, and
$s\pijl{}{}B$ iff there is some $s'\in B$ such that $s\pijl{}{}s'$. We write 
$s\npijl{}s'$ and $s\npijl{}B$ iff it is not the case that $s\pijl{}{}s'$, resp., $s\pijl{}{}B$.

\begin{definition}{Divergence-blind stuttering equivalence}
Let $K=(S,\mathit{AP},\pijl{}{},L)$ be a Kripke structure.
A symmetric relation $R\subseteq S\times S$ is a \textit{divergence-blind stuttering
equivalence} iff for all $s,t\in S$ such that $sRt$:
\begin{compactenum}
\item
$L(s)=L(t)$.
\item
for all $s'\in S$ if $s\pijl{}{}s'$, then there are $t_0,\ldots, t_k\in S$ for some $k\in\Nat$
such that $t=t_0$, $sRt_i$, $t_i\pijl{}{}t_{i+1}$, and $s'Rt_k$ for all $i<k$.
\end{compactenum}
We say that two states $s,t\in S$ are \textit{divergence-blind stuttering equivalent},
notation $s{\bis_{\mathit{dbs}}} t$, iff there is a divergence-blind stuttering
equivalence relation $R$ such that $sRt$.
\end{definition}
An important property of divergence-blind stuttering equivalence is that if states on a loop
all have the same label then all these states are divergence-blind stuttering equivalent. 
%Another important property of bisimulations is that if two states are bisimilar, then they
%can be joined,
%merging the incoming and outgoing transitions,
%obtaining a minimal transition system.
We define stuttering equivalence in terms of divergence-blind stuttering equivalence 
using the following Kripke structure.
\begin{definition}{Stuttering equivalence}
Let $K=(S,\mathit{AP},\pijl{}{},L)$ be a Kripke structure. Define the
Kripke structure $K_d=(S\cup\{s_d\},\mathit{AP}\cup\{d\},\pijl{}{d},L_d)$ where $d$ is an atomic
proposition not occurring in $\mathit{AP}$ and $s_d$ is a fresh state not occurring in $S$.
Furthermore,
\begin{compactenum}
\item
$\pijl{}{d}=\pijl{}{}\cup\{\langle s,s_d\rangle~|~s$ is on a cycle of states all labelled with $L(s)$, or $s=s_d\}$.
\item
For all $s\in S$ we define $L_d(s)=L(s)$ and $L_d(s_d)=\{d\}$.
\end{compactenum}
States $s,t\in S$ are \textit{stuttering equivalent}, 
notation $s{\bis_s} t$ iff there is a divergence-blind
stuttering equivalence relation $R$ on $S_d$ such that $sRt$.
\end{definition}
Note that an algorithm for divergence-blind stuttering equivalence
can also be used to determine stuttering equivalence by employing only a linear
time and space transformation. Therefore, we only concentrate on an algorithm
for divergence-blind stuttering equivalence.
\section{Partitions and splitters: a simple algorithm}
Our algorithms perform partition refinement of an initial partition containing the set of states $S$. 
A \textit{partition} $\pi=\{ B_i\subseteq S~|~1\leq i\leq k\}$ is a set of non empty 
subsets such that $B_i\cap B_j=\emptyset$ for all $1\leq i<j\leq k$ and $S =\bigcup_{1\leq i\leq k}B_i$. 
Each $B_i$ is called a \textit{block}.
% of the partition.

We call a transition $s \pijl{}{} s'$ \textit{inert w.r.t.\ $\pi$} iff $s$ and $s'$ are in the same block $B \in \pi$.
We say that a partition $\pi$ \textit{coincides} with divergence-blind stuttering equivalence 
when $s{\bis_{\mathit{dbs}}} t$ iff there is a block $B\in \pi$ such that $s,t\in B$. We say that 
a partition \textit{respects} divergence-blind stuttering equivalence iff for all $s,t\in S$ 
if $s{\bis_{\mathit{dbs}}} t$ then there is some block $B\in \pi$ such that $s,t\in B$. The goal
of the algorithm is to calculate a partition that coincides with divergence-blind stuttering
equivalence. This is done starting with the initial partition $\pi_0$ consisting of blocks
$B$ satisfying that if $s,t\in B$ then $L(s)=L(t)$. Note that this initial partition respects 
divergence-blind stuttering equivalence. 

We say that a partition $\pi$ is \textit{cycle-free} iff there is 
no state $s\in B$ such that $s\pijl{}{B}s_1$ $\pijl{}{B}\cdots\pijl{}{B}s_k\pijl{}{}s$ for some $k\in\Nat$
for each block $B\in \pi$. 
It is easy to make the initial partition $\pi_0$ cycle-free by merging all states on a cycle in each block into a single state.
This preserves divergence-blind stuttering equivalence and can be performed in linear time employing a standard algorithm
to find strongly connected components~\cite{H74}. 

The initial partition is refined until it 
coincides with divergence-blind stuttering equivalence.
Given a block $B'$ of the current partition and the union $\setB$ of some of the blocks in the partition, we define
%\begin{shrinkeq}{-1ex}
\[\begin{array}{l}
\splitpi(B',\setB)=\{ s_0{\in} B' \mid \exists k {\in} \Nat, s_1,.., s_k {\in} S. s_i\pijl{}{}s_{i+1}, s_i{\in} B'\textrm{ for all }i<k \wedge s_k  {\in} \setB \}
\\
\cosplitpi(B',\setB)=B'\setminus\splitpi(B',\setB).
\end{array}\]
%\end{shrinkeq}
Note that if $B'\subseteq \setB$, then $\splitpi(B',\setB)=B'$.
The sets $\splitpi(B',\setB)$ and $\cosplitpi(B',\setB)$ are intended as the
new blocks to replace $B'$.
It is common to split blocks under single blocks, i.e., $\setB$ corresponding with a single block $B\in \pi$~\cite{GV90,KS90}.
However, as indicated in \cite{PT87}, it is required to split under the union of some of the blocks in $\pi$ to 
obtain an $O(m \log n)$
algorithm. We refer to such unions as \textit{constellations}. In section~\ref{sec:constellations}, we use 
constellations consisting of more than one block in the splitting.

We say that a block $B'$ is \textit{unstable} under $\setB$ iff
$\splitpi(B',\setB)\not=\emptyset$ and $\cosplitpi(B',\setB)$ $\not=\emptyset$. 
A partition $\pi$ is \textit{unstable} under $\setB$ iff there is at least one $B' \in \pi$
which is unstable under $\setB$.
If $\pi$ is 
not unstable under $\setB$ then it is called \textit{stable under} $\setB$.
If $\pi$ is stable under all $\setB$, then it is simply called stable.

A \textit{refinement} of $B' \in \pi$ under $\setB$
consists of two new blocks $\splitpi(B',\setB)$ and $\cosplitpi(B',\setB)$.
A partition $\pi'$ is a refinement of $\pi$ under $\setB$ iff all unstable blocks $B' \in \pi$
have been replaced by new blocks $\splitpi(B',\setB)$ and $\cosplitpi(B',\setB)$.

The following lemma expresses that if a partition is stable then it coincides with divergence-blind stuttering equivalence.
It also says that during refinement, the encountered partitions respect divergence-blind stuttering equivalence 
and remain cycle-free.
\begin{lemma}
Let $K=(S,\mathit{AP},\pijl{}{}, L)$ be a Kripke structure and $\pi$ a partition of $S$.
\begin{enumerate}
\item
For all states $s,t\in S$, if $s,t\in B$ with $B$ a block of the partition $\pi$, $\pi$ is stable, and
a refinement of the initial partition $\pi_0$, then $s{\bis_\mathit{dbs}} t$.
\item
If $\pi$ respects divergence-blind stuttering equivalence 
then any refinement of $\pi$ under the union of some of the blocks in $\pi$ also respects it.
\item
If $\pi$ is a cycle-free partition, then any refinement of $\pi$ is also cycle-free.
\end{enumerate}
\end{lemma}
%\begin{proof}
%See Appendix A.
%\vspace{-3ex}
%\end{proof}
\begin{proof}
\begin{enumerate}
\item
We show that if $\pi$ is a stable partition, the relation $R=\{\langle s,t\rangle~|~s,t\in B,~B\in\pi\}$ is
a divergence-blind stuttering equivalence. It is clear that $R$ is symmetric. Assume $sRt$. Obviously, $L(s)=L(t)$ 
because $s,t\in B$ and $B$ refines the initial partition. For the second requirement of divergence-blind
stuttering equivalence, suppose $s\pijl{}{} s'$. There is a block $B'$ such
that $s'\in B'$. As $\pi$ is stable, it holds for $t$ that $t=t_0\pijl{}{}t_1\pijl{}{}\cdots \pijl{}{} t_k$ for
some $k\in\Nat$, $t_0,\ldots ,t_{k-1}\in B$ and $t_k\in B'$. This clearly shows that for all $i<k$ $sRt_i$, and $s'Rt_k$. 
So, $R$ is a divergence-blind stuttering equivalence, and therefore 
it holds for all states $s,t\in S$ that reside in the
same block of $\pi$ that $s{\bis_\mathit{dbs}}t$.
\item
The second part can be proven by reasoning towards a contradiction. 
Let us assume that a partition 
$\pi'$ that is a refinement of $\pi$ under $\setB$ does not
respect divergence-blind 
stuttering equivalence, although $\pi$ does.
Hence, there are states $s, t\in S$ with $s {\bis_{\mathit{dbs}}} t$ and
a block $B' \in \pi$ with $s, t\in B'$ and $s$ and $t$ are in different blocks in $\pi'$. Given that $\pi'$ is a 
refinement of $\pi$ under $\setB$, $s \in \splitpi(B',\setB)$ and $t \in \cosplitpi(B', \setB)$ (or vice versa,
which can be proven similarly). By definition of $\splitpi$, there are 
$s_1, \ldots, s_{k-1} \in B'$ ($k \in \Nat$) and $s_k \in \setB$ such that $s \pijl{}{} s_1 \pijl{}{} \cdots \pijl{}{} 
s_k$. Then, either $k=0$ and $B'\subseteq \setB$, but then $t\notin\cosplitpi(B',\setB)$. Or $k>0$, and 
since $s {\bis_{\mathit{dbs}}} t$, there are 
$t_1, \ldots, t_{l-1} \in B'$ ($l \in \Nat$) and $t_l \in \setB$ such that 
$t \pijl{}{} t_1 \pijl{}{} \cdots \pijl{}{} t_l 
$ with $s_i R t_j$ for all $1 \leq i < k$, $1 \leq j < l$ and $s_k R t_l$. This means that we have 
$t \in \splitpi(B', \setB)$, again contradicting that $t \in \cosplitpi(B', \setB)$.
\item
If $\pi$ is cycle-free, this property is straightforward, since splitting any block of $\pi$ will not introduce cycles.
%\qed
\end{enumerate}
%\vspace{-3ex}
\end{proof}

This suggests the following simple algorithm which has time complexity $O(m n)$ and space complexity $O(m{+}n)$,
which essentially was presented in \cite{GV90}.
\begin{center}
%\scalebox{0.8}{
\begin{tabular}{|l|}
\hline
\hspace{1cm}$\pi:=\pi_0$, i.e., the initial partition;\\
\hspace{1cm}\textbf{while} $\pi$ is unstable under some $B\in \pi$ \hspace*{1cm}\\
\hspace{1.5cm}$\pi:=$ refinement of $\pi$ under $B$;\\
\hline
\end{tabular}
%}
\end{center}
\noindent
It is an invariant of this algorithm that $\pi$ respects 
divergence-blind stuttering equivalence and $\pi$ is cycle-free. In particular, $\pi=\pi_0$ satisfies this invariant initially.
If $\pi$ is not stable, a refinement under some block $B$ exists, splitting at least one block.
Therefore, this algorithm finishes in at most $n{-}1$ steps as during each iteration of the algorithm
the number of blocks increases by one, and the number of blocks can never exceed the number of states. 
When the algorithm terminates, $\pi$ is stable and therefore,
two states are divergence-blind stuttering equivalent iff they are part of the same block in the
final partition. This end result is independent of the order in which splitting took place.

In order to see that the time complexity of this algorithm is $O(m n)$, we must show that we can detect that $\pi$
is unstable and carry out splitting in time $O(m)$.
The crucial observation to efficiently determine whether a partition is stable stems from \cite{GV90} where
it was shown that it is enough to look at the bottom states of a block, which always exist for each block
because the partition is cycle-free.
The \textit{bottom states} of a block
are those states that do not have an outgoing inert transition, i.e., a transition 
to a state in the same block. They are defined by
\[\mathit{bottom}(B)=\{s\in B~|~\textrm{there is no state }s'\in B\textrm{ such that }s\pijl{}{}s'\}.\]

The following lemma presents the crucial observation concerning bottom states.
%, except that it is formulated more generally for sets of states, but we require that later.
\begin{lemma}
\label{branchingsplitlemma}
Let $K=(S,\mathit{AP},\pijl{}{},L)$ be a Kripke structure and 
$\pi$ be a cycle-free partition of its states.
Partition $\pi$ is unstable under union $\setB$ of some of the blocks in $\pi$ iff there is a block $B'\in\pi$ such that
\[\emptyset\subset \splitpi(B',\setB) \textrm{ and }
\mathit{bottom}(B')\cap\splitpi(B',\setB)\subset \mathit{bottom}(B').\]
Here $\subset$ is meant to be a strict subset.
\end{lemma}

\begin{proof}
\begin{itemize}
\item[$\Rightarrow$]
If $\pi$ is unstable, then $\splitpi(B',\setB)\not=\emptyset$ and $\splitpi(B',\setB)\not= B'$.
The first conjunct immediately implies $\emptyset \subset \splitpi(B',\setB)$. 
If $\splitpi(B',\setB)\not= B'$, there are states $s {\notin} \splitpi(B',\setB)$. 
As the blocks $B' {\in} \pi$ do not have cycles, consider such an $s \notin \splitpi(B', \setB)$ with a smallest distance to 
a state $s_k {\in} \mathit{bottom}(B')$, i.e., $s\pijl{}{}s_1\pijl{}{}\cdots$ $\pijl{}{}s_k$ with all $s_i\in B'$.
If $s$ itself is an element of $\mathit{bottom}(B')$, the second part of the right hand side of the lemma follows.
Assume $s{\notin} \mathit{bottom}(B')$, there is some state $s'{\in} B'$ closer to $\mathit{bottom}(B')$ such 
that $s\pijl{}{}s'$. Clearly, $s'{\notin} \splitpi(B',\setB)$ either, as otherwise $s\in \splitpi(B',\setB)$.
But as $s'$ is closer to $\mathit{bottom}(B')$, the state $s$ was not a state with the smallest distance 
to a state in $\mathit{bottom}(B')$, which is a contradiction. 
\item[$\Leftarrow$]
It follows from the right hand side that
$\splitpi(B',\setB)\not=\emptyset$, $\splitpi(B',\setB)\not=B'$.
%\qed
\end{itemize}
%\vspace{-4ex}
\end{proof}
This lemma can be used as follows to find a block to be split. Consider each $B \in \pi$.
Traverse its incoming transitions and mark the states that can reach $B$ in zero or one step.
%If the source state is not in $B$ only one step applies.
If a block $B'$ has marked
states, but not all of its bottom states are marked, the condition of the lemma applies, and it needs to be split.
It is at most needed to traverse all transitions to carry this out, so its complexity is
$O(m)$.

If $B$ is equal to $B'$, no splitting is possible. We implement it by marking all states in $B$ as each state in $B$
can reach itself in zero steps. In this case condition $\mathit{bottom}(B')\cap\splitpi(B',\setB)\subset \mathit{bottom}(B')$
is not true. This is different from \cite{GV90} where a block is never considered as a splitter of itself,
but we require this in the algorithm in the next sections. 

If a block $B'$ is unstable, and all states from which a state in $B$ can be reached in one step are marked,
then a straightforward recursive procedure is required to extend the marking to all states in $\splitpi(B',B)$,
and those states need to be moved to a new block. This takes time proportional to 
the number of transitions in $B'$, i.e., $O(m)$.

\section{Constellations: an $O(m\log n)$ algorithm}
\label{sec:constellations}
The crucial idea to transform the algorithm from the previous section into an $O(m\log n)$ algorithm stems
from \cite{PT87}. By grouping the blocks in the current partition $\pi$ into constellations such that $\pi$ is stable
under the union of the blocks in such a constellation, 
we can determine whether a block exists under which $\pi$ is unstable by only looking at blocks that are at most half
the size of the constellation, i.e., $|B| \leq \frac{1}{2}\Sigma_{B'{\in \setB}} |B'|$, 
for a block $B$ in a constellation $\setB$. 
If a block $B'\in\pi$
is unstable under $B$, then we use a remarkable technique consisting of two procedures running 
alternatingly to identify the smallest
block resulting from the split. The whole operation runs
in time proportional to the
smallest block resulting from the split. We involve the blocks in $\setB \setminus B$\footnote{For convenience,
we write $\setB \setminus B$ instead of $\setB \setminus \{ B \}$.} in the splitting without explicitly
analysing the states contained therein.

Working with constellations in this way ensures for each state that whenever it is involved in splitting, i.e., if it is part of a block that is
used to split or that is being split, this block
is half the size of the previous block in which the state resided when it was involved in splitting. That ensures that each state can 
at most be $\log_2(n)$ times involved in splitting. When involving a state, we only analyse its incoming and outgoing
transitions, resulting in an algorithm with complexity $O(m\log n)$. Although we require quite a number
of auxiliary data structures, these are either proportional to the number of states or to the number of transitions.
So, the memory requirement is $O(m{+}n)$.

% split a block $B$ under a block $C$ in some constellation \setC in time linear to the size of the smallest block resulting from the split. Also,

%There are two aspects that make the new algorithm much more complex than the sketch above suggests. The first one has to do with the problem that splitting a block $B$ must be done in time proportional to the size of the smallest resulting block, which will have a size smaller than or equal to half the size of $B$. Also, when we check whether the current partition is stable 

%We call $B$ the \textit{constellation} in which $B_1$ and $B_2$ reside.
%We need to check whether the current partition is stable with respect to
%$B_1$ and $B_2$. This must be resolved in time proportional to $B_1$ (as is done in \cite{PT87}). 
%
%When investigating $B_1$, we need to know its constellation $B$.

In the following, the set of constellations also forms a partition, 
which we denote by $\cal C$. 
A constellation is the union of one
or more blocks from the current partition. If it corresponds with one block, the constellation is called trivial.
The current partition is stable with respect to each constellation in ${\cal C}$. 

If a constellation \setB contains more than one block, we select one block $B {\in} \setB$ which is at most
half the size of \setB. We check whether the current partition is stable under $B$ and $\setB\setminus B$
according to lemma \ref{branchingsplitlemma} by traversing the incoming transitions of states in $B$ and marking
the encountered states that can reach $B$ in zero or one step.
For all blocks $B'$ that are unstable according to lemma \ref{branchingsplitlemma}, we calculate $\splitpi(B',B)$ and 
$\cosplitpi(B',B)$, as indicated below.

As noted in~\cite{PT87}, $\cosplitpi(B',B)$ is stable under $\setB \setminus B$. Therefore,
only further splitting of $\splitpi(B',B)$ under $\setB \setminus B$ must be investigated. If $B'$ is stable under $B$ 
because all bottom states of $B'$ are marked, 
it can be that $B'$ is not stable under $\setB\setminus B$, which we do not address here explicitly, as it proceeds along the
same line.

There is a special data structure to recall for any $B'$ and
$\setB$ which transitions go from $B'$ to $\setB$. When investigating whether $\splitpi(B',B)$ is stable
under $B$ we adapt this list to determine the transitions from $\splitpi(B',B)$ to $\setB\setminus B$ and
we simultaneously tag the states in $B'$ that have a transition to $\setB\setminus B$. 
Therefore, we know whether there are transitions from $\splitpi(B',B)$ to $\setB\setminus B$ and
we can traverse the bottom states of $\splitpi(B',B)$ to inspect whether there is a bottom state without a transition to
$B$. Following lemma \ref{branchingsplitlemma}, 
this allows us to determine whether $\splitpi(B',B)$ must be split under $\setB\setminus B$ in a
time proportional to the size of $B$. How splitting is carried out is indicated below.

When the current partition has become stable under $B$ and $\setB\setminus B$, $B$ is
moved from constellation $\setB$ into a new trivial constellation $\setBp$, and the constellation
$\setB$ is reduced to contain the states in $\setB\setminus B$. Note that the new $\setB$ can have become trivial.

There is one aspect that complicates matters. If blocks are split, the new partition is not
automatically stable under all constellations. This is contrary to the situation in \cite{PT87} and was already observed
in \cite{GV90}. Figure~\ref{fig:unstable} indicates the situation.
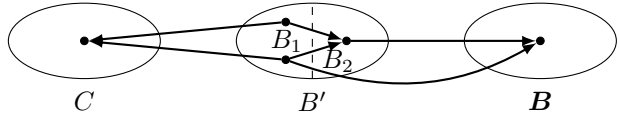
\begin{figure}[t]
\begin{center}
\scalebox{1.0}{
\begin{tikzpicture}
\draw (0,0.5) ellipse (1cm and 0.5cm);
\draw (0,-0.3)  node (D) {$C$};
\draw (3,0.5) ellipse (1cm and 0.5cm) node[below right] (B2) {$B_2$} node[left] (B1) {$B_1$};
\draw (3,-0.3) node (B) {$B'$};
\draw (6,0.5) ellipse (1cm and 0.5cm);
\draw (6,-0.3) node (C) {$\setB$};
\draw[dashed] (3,0) --  (3,1);
\draw (2.65,0.75) node[circle, minimum size=1mm, fill, draw, inner sep=0pt] (n1) {};
\draw[fill] (2.65,0.25) node[circle, minimum size=1mm, fill, draw, inner sep=0pt] (n2) {};
\draw[fill] (3.45,0.5) node[circle, minimum size=1mm, fill, draw, inner sep=0pt] (n3) {};
\draw[fill] (0,0.5) node[circle, minimum size=1mm, fill, draw, inner sep=0pt] (n4) {};
\draw[fill] (6,0.5) node[circle, minimum size=1mm, fill, draw, inner sep=0pt] (n5) {};
\draw[thick,-latex] (n1) -- (n3);
\draw[thick,-latex] (n2) -- (n3);
\draw[thick,-latex] (n3) -- (n5);
\draw[thick,-latex] (n2) to[out=-20, in=210] (n5);
\draw[thick,-latex] (n1) -- (n4);
\draw[thick,-latex] (n2) -- (n4);
\end{tikzpicture}
}
\end{center}
\caption{After splitting $B'$ under $C$, $B_1$ is not stable under $\setB$.}
\label{fig:unstable}
\end{figure}
Block $B'$ is stable under constellation $\setB$. But if $B'$ is split under block $C$
into $B_1$ and $B_2$, block
$B_1$ is not stable under $\setB$. The reason is, as exemplified by the following lemma,
that some states that were non-bottom states in $B'$ became bottom states in $B_1$.

\begin{lemma}
Let $K=(S,\mathit{AP},\pijl{}{},L)$ be a Kripke structure with cycle free partition $\pi$ with refinement $\pi'$.
If $\pi$ is stable under a constellation $\setB$, and $B'\in \pi$ is refined into $B_1',\ldots,B_k'\in\pi'$, then for
each $B_i'$ where the bottom states in $B_i'$ are also bottom states in $B'$, it holds that $B_i'$ is also
stable under $\setB$.
\end{lemma}
\begin{proof}
Assume $B_i'$ is not stable under $\setB$. This means that $B_i'$ is not a subset of $\setB$.
Hence, there is a state $s\in B_i'$ such that
$s\pijl{}{} s'$ with $s' \in \setB$ and there is a bottom state $t\in B_i'$ with no outgoing transition to a state
in $\setB$. But
as $B'$ was stable under $\setB$, and $s$ has an outgoing transition
to a state in $\setB$, all bottom states in $B'$ must have at least one transition to a state in $\setB$.
Therefore, $t$ cannot be a bottom state of $B'$, and must have become a bottom state after splitting $B'$.
%\qed
\end{proof}

This means that if a block $B'$ is the result of a refinement, and some of its
states became bottom states, it must be made sure that $B'$ is stable under the
constellations. Typically, from the new bottom states a smaller number of blocks
in the constellation can be reached. For each block we maintain a list of constellations 
that can be reached from states in this block. We match the outgoing transitions of the new bottom states
with this list, and if there is a difference, we know that $B'$ must be split further. 

The complexity of checking for additional splittings to regain stability when states become bottom states
is only $O(m)$. Each state only becomes a bottom state once, and when that happens we perform calculations
proportional to the number of outgoing transitions of this state to determine whether a split must be carried out.

It remains to show that splitting can be performed in a time proportional to the size of the
smallest block resulting from the splitting. Consider splitting $B'$ under $B{\in}\setB$.
While marking $B'$ four lists of all marked and non marked, bottom and non bottom 
states have been constructed. We simultaneously mark states in $B'$ either red or blue. Red means
that there is a path from a state in $B'$ to a state in $B$. Blue means that there is no such path.
Initially, marked states are red, and non marked bottom states are blue. 

This colouring is simultaneously extended to all states in $B'$, spending equal time to both. 
The procedure is stopped when the colouring of one of the colours cannot be enlarged.
We colour states red that can reach other red states via inert
transitions using a simple recursive procedure. We colour states blue
for which it is determined that all outgoing inert transitions go to a blue state 
(for this we need to recall for each state the number
of outgoing
inert transitions) and there is no direct transition to $B$. 
The marking procedure that terminates first, provided that its number of marked states
does not exceed $\frac{1}{2}|B'|$, has the smallest block that must be split.
Now that we know the smallest
block we move its states to a newly created block. 

Splitting regarding $\setB\setminus B$ only has to be applied to $\splitpi(B', B)$, or to $B'$ if all
bottom states of $B'$ were marked.
As noted before $\cosplitpi(B', B)$ is stable under $\setB \setminus B$. Define $C :=\splitpi(B',B)$ or
$C:=B'$ depending on the situation.
We can traverse all bottom states of $C$ and check whether they have outgoing transitions
to $\setB\setminus B$. This provides us with the blue states. 
The red states are obtained as we explicitly maintained the list of all transitions from $C$ to 
$\setB\setminus B$.
By simultaneously extending this colouring the smallest subblock of either red or blue states is 
obtained and splitting can commence.

The algorithm is concisely presented below. After that, it is presented
in full detail in section~\ref{sec:detalg}. Since it is not trivial how to achieve the $O(m\log n)$ complexity, we have decided
to describe the algorithm as detailed as possible.
\begin{center}
%\scalebox{0.8}{
\begin{tabular}{|l|}
\hline
\hspace{0.5cm}$\pi:=$ initial partition; ${\cal C}:=\{\pi \};$\\
%\hspace{0.5cm}${\cal C}:=\{S \};$\\
\hspace{0.5cm}\textbf{while} $\cal C$ contains a non trivial constellation $\setB\in {\cal C}$\\
\hspace{1.0cm}\textbf{choose} some $B\in\pi$ such that $B \in\setB$ and $|B|\leq \frac{1}{2}|\setB|$;\\
\hspace{1.0cm}${\cal C}:=$partition ${\cal C}$ where $\setB$ is replaced by $B$ and $\setB\setminus B$;\\

%\hspace{1.0cm}${\cal C}:=({\cal C}\setminus \setB)\cup \{\setB\setminus\{ B\}\}\cup \{\{B\}\};$\\
\hspace{1.0cm}\textbf{if} $\pi$ is unstable for $B$ or $\setB\setminus B$\\
\hspace{1.5cm}$\pi':=$ refinement of $\pi$ under $B$ and $\setB\setminus B$;\\
%\hspace{2cm}${\cal C}:=
%\textrm{refinement of all constellations in }{\cal C}\textrm{ under }B\textrm{ and }\setB\setminus B$\\
%\hspace{2cm}${\cal C}:={\cal C}\cup \{\textrm{refinement of }C\in\pi\textrm{ under }B\textrm{ and }\setB\setminus B~
%|~C\notin \setB\textrm{ for all }\setB\in {\cal C}$\hspace*{1cm}\\
%\hspace{4cm}$\textrm{ and }C\textrm{ is unstable under }B\textrm{ and }\setB\setminus B\}$\\

\hspace{1.5cm}For each block $C\in \pi'$ with bottom states that were not bottom in $\pi$\hspace*{0.5cm}\\
\hspace{3.5cm}split $C$ until it is stable for all constellations in $\cal C$;\\
\hspace{1.5cm}$\pi:=\pi'$\\
\hline
\end{tabular}
%}
\end{center}

\section{Detailed algorithm}
\label{sec:detalg}
This section presents the data structures, the
algorithm to detect which blocks must be split, and the algorithm to split blocks. 
It follows the 
outline presented in the previous section. 
%Observe that we have to store a substantial amount of 
%auxiliary, but redundant data required, such as $s.\inertcounter$ counting the number of inert transitions
%that leave a state, or $(s\pijl{}{}s').\blockconstellationlist$ that contains all transitions from the 
%block containing $s$ to the constellation containing $s'$. All this data 
%is needed to obtain required information within the very limited amount of steps that the $O(m\log n)$ 
%algorithm allows us to use. Maintaining this data while splitting blocks is a tedious job.
\subsection{Data structures}

%We start explaining the data structures that we need for the algorithm.
As a basic data structure, we use (singly-linked) lists. For a list $L$ of elements, we assume that for each element $e$, a reference
to the position in $L$ preceding the position of $e$ is maintained, such that checking membership and removal can be done in constant time. 
In some cases we add some extra information to the elements in the list.
Moreover, for each list $L$, we maintain pointers to its first and last element, and the size $|L|$.

\begin{enumerate}
\item
The current partition $\pi$ consists of a list of blocks. Initially, it corresponds with $\pi_0$. All blocks are part of a single initial
constellation $\setC_0$.
\item
For each block $\B$, we maintain the following:
\begin{enumerate}
\item A reference $\B.\constellation$ to the constellation containing $\B$.
\item A list of the bottom states in \B called \B.\bottomstates.
\item A list of the remaining states in \B called \B.\nonbottomstates.
\item A list $B.\toconstellations$ of structures associated with constellations reachable via a transition from some $s {\in} B$. 
Initially, it contains one element associated with $\setC_0$. Each element associated with some constellation
$\setC$ in this list also contains the following:
\begin{itemize}
\item A reference $\transitionlist$ to a
list of all transitions from states in $B$ to states in $\setC \setminus B$
(note that transitions between states in $B$, i.e., inert transitions, are \emph{not} in this list).
\item When splitting the block $B$ into $B$ and $B'$ 
there is a reference in each list element to the corresponding list element in  $B'.\toconstellations$ (which in turn refers back to the element in $B.\toconstellations$).
\item In order to check for stability when splitting produces new bottom states, each element contains
a list to keep track of which new bottom states can reach the associated constellation. 
\end{itemize}
\item A reference $B.\inconstellationref$ is used to refer to the element in $B.\toconstellations$ associated with
constellation $B.\constellation$.
\end{enumerate}

Furthermore, when splitting a block $B'$ in constellation $\setB'$ under a constellation \setB and block $\B {\in} \setB$, the 
following temporary structures are used, with $\setC$ the new constellation to which $\B$ is moved:
\begin{enumerate}
\item A list $B'.\markedbottomstates$ (initially empty) contains states in $B'$ with a transition to $\B$.
\item A list $B'.\markednonbottomstates$ (initially empty) contains states that are marked, but are not bottom states, i.e., 
each of those states has at least one transition to $\B$ and at least one transition to $B'$.
\item A reference $B'.\constellationref$ is used to refer to the (new) element in $B'.\toconstellations$ associated with
constellation $\setC$, i.e., the new constellation of $B$.
\item A reference $B'.\coconstellationref$ is used to refer to the element in $B'.\toconstellations$
associated with constellation $\setB$, i.e., the old constellation of $B$.
\item A list $B'.\newbottomstates$ to keep track of the states that have become bottom states when $B'$ was split. This is required to determine whether $B'$ is stable under all constellations after a split.
%\item A reference $B'.\toconstellationelement$ to the constellation $\setB$ in $B'.\toconstellations$.
%\item A pointer $$B'$.\movepointer$ (initially \NULL), to be used when moving blocks between constellations.
\end{enumerate}
%NOT NEEDED?:, and a counter $$B'$.\constellationtransitioncounter$, which contains the number of $a$-transitions from $B'$ to $\setB\setminus \B$.
\item
The constellations are represented by two lists. The first one, \nontrivialconstellations, contains
the constellations encompassing two or more blocks of the current partition.
The second one, \trivialconstellations, contains constellations that match a block of the current partition.
Initially, if the initial partition $\pi_0$ consists of one block,
the constellation $\setC_0 = \{ \pi_0 \}$ is added to \trivialconstellations and nothing needs to be done, because the initial partition
is already stable. Otherwise $\setC_0$ is added to \nontrivialconstellations.
\item
For each constellation \setB, the following is maintained:
\begin{enumerate}
\item 
A list $\setB.\blocks$ of blocks contained in $\setB$. Initially, the only constellation is $\setC_0$.
\item Counter $\setB.\size$ is used to keep track of the number of states in $\setB$; it is equal to $\Sigma_{B{\in}\setB.\blocks} |B|$.
%\item 
%A list $\setB.\fromblocks$ containing triples $(B', \transitionsinblocktoconstellation, p)$, 
%where $B'$ is a block from which at least one transition exists leading to a state in $
%\setB$, $\transitionsinblocktoconstellation$ is the list of transitions from $B'$ to $\setB$, 
%and $p$ is a pointer to the corresponding entry of $\setB$ in $B'.\toconstellations$.
%\item For each block $B'$ in $\setB.\blockstoconstellation$, a list $B'.\statescounttoconstellation$ of \textit{(state, counter)} pairs indicating for each state $s \in B'$ how many outgoing transitions $s$ has to a state in $\setB$. In relation to these counters, each transition 
%$s \pijl{}{} s'$ and $s' \naarpijl{} s$ has two pointers, with $s' \in \setB$:
%\begin{enumerate}
%\item One pointer to an entry $(s, n)$ in $B'.\statescounttoconstellation$ of constellation 
%$\setB$, to keep track of how many outgoing transitions $s$ has to the same constellation as $s'$. We refer with $(s \pijl{}{} s')\deref \toconstellationcounter$ to the associated counter.
%\item One pointer to the list $B'.\statescounttoconstellation$ of constellation $\setB$ as a whole. We refer with $(s \pijl{}{} s') \deref \statescounttoconstellation$ to that list, and therefore with $|(s \pijl{}{} s') \deref \statescounttoconstellation|$ to the size of that list, corresponding to the number of states in $B'$ with at least one transition to $\setB$.
%\end{enumerate}

%\item Temporary pointers $\setB.\movepointer$ and $\setB.\movestatecounterpointer$ (both initially \NULL), to be used when moving states between blocks.
\end{enumerate}
%\item
%There is also a list of transitions $T_{\cal C}$ from blocks to constellations, i.e., with transitions of the form $\B\pijl{}{}\setB$.
%For each such transition, the counter $(\B \pijl{}{} \setB).\numberofstatescounttoconstellation$ contains the number of states in \B that have at least one transition to \setB. A transition $s \pijl{}{} s'$ is associated iff $s \in \B$ and $s' \in \unionsetB$. Each such transition has a reference to its associated transition 
%$\B\pijl{}{}\setB$. 

%In addition, when splitting a block $B'$ producing a new block $B'$p, we use for every transition $$B'$ \pijl{a}{} \setB$ a temporary variable $\cotransition$ containing a new transition $$B'$p \pijl{a}{} \setB$ iff there is at least one $s \in $B'$p$ with an $a$-transition to a state $s' \in \unionsetB$. Also this transition will have a counter $\numberofbottomstatescounttoconstellation$ to count the number of bottom states with at least one $a$-transition to $\setB$.
\item
Each transition $s\pijl{}{}s'$ contains its source and target state.
Moreover, it refers with $\toconstellationcount$ to a variable containing the number of
transitions from $s$ to the constellation in which $s'$ resides. For each state and constellation, there is one such variable, provided there is a transition from this state to this constellation.

Each transition $s\pijl{}{}s'$ has a reference to the element associated with $\setB$ in the list $B.\toconstellations$
where $s {\in} B$ and $s'{\in}\setB$. This is denoted as $(s\pijl{}{}s').\toconstellationref$.
Initially, it refers to the single element in $B.\toconstellations$, unless the transition is
inert, i.e., both $s{\in}B$ and $s'{\in}B$.

Furthermore, each transition $s\pijl{}{}s'$ is stored in the list of transitions from $B$ to $\setB$.
Initially, there is such a list 
for each block in the initial partition $\pi_0$. From a transition $s \pijl{}{} s'$, the list can be accessed via $(s\pijl{}{}s').\toconstellationref.\transitionlist$.

\item
For each state $s {\in} \B$ we maintain the following information:
\begin{enumerate}
\item A reference $s.\block$ to the block containing $s$.
\item
A static, i.e., not changing during the course of the algorithm, list $s.T_{\mathit{tgt}}$ of transitions of the form $s \pijl{}{} s'$
containing precisely all the transitions from $s$. 
%With each transition $s \pijl{}{} s'$ in this list, where $s\in\setBp$,
%a pointer is associated to a tuple $(\toconstellationcount, p)$, where $\toconstellationcount \in \Nat$ is the number of outgoing %transitions of $s$ that lead to $\setBp$, and $p$ is a pointer to the corresponding entry for $B$ in $\setBp.\fromblocks$. 
%Furthermore, with each transition $s \pijl{}{} s'$ in this list, two pointers are associated:
%\begin{enumerate}
%\item $(s \pijl{}{} s').\blockconstellationpointer$, pointing to the entry in $\setBp.\blockstoconstellation$ corresponding with block \B, where $s' \in \unionsetBp$.
%\item $(s \pijl{}{} s').\statescounterpointer$, pointing to the counter associated with $s$ in the list of counters associated with $(s \pijl{}{} s').\blockconstellationpointer$.
%\end{enumerate}
\item
A static list $s.T_{\mathit{src}}$ of transitions $s' \pijl{}{} s$ containing all the transitions to $s$. 
In the sequel we write such transitions as $s \naarpijl{} s'$, to stress that these move into $s$.
%As with the transitions in $s.T_{\mathit{tgt}}$, with each transition $s \naarpijl{} s'$, a pointer is associated to a tuple $
%(\toconstellationcount, p)$.
%Furthermore, with each transition $s \naarpijl{} s'$ in this list, two pointers are associated:
%\begin{enumerate}
%\item $(s \naarpijl{} s').\blockconstellationpointer$, pointing to the entry in $\setB.\blockstoconstellation$ corresponding with block $B'$, where $s' \in $B'$$, $s \in \unionsetB$.
%\item $(s \naarpijl{} s').\statescounterpointer$, pointing to the counter associated with $s'$ in the list of counters associated with $(s \naarpijl{} s).\blockconstellationpointer$.
%\end{enumerate}
\item A counter $s.\inertcounter$ containing the number of outgoing transitions to a state in the same block as $s$. 
For any bottom state $s$, we have $s.\inertcounter = 0$.
%\end{enumerate}

\item Furthermore, when splitting a block $B'$ under \setB and $\B {\in} \setB$, there are 
references 
$s.\constellationcounter$ and $s.\coconstellationcounter$ to the variables that are used to count how many transitions
there are from $s$ to $B$ and from $s$ to $\setB\setminus \B$.
%\item A pointer $s.\movepointer$ (initially \NULL), to be used when moving blocks between constellations.
\end{enumerate}

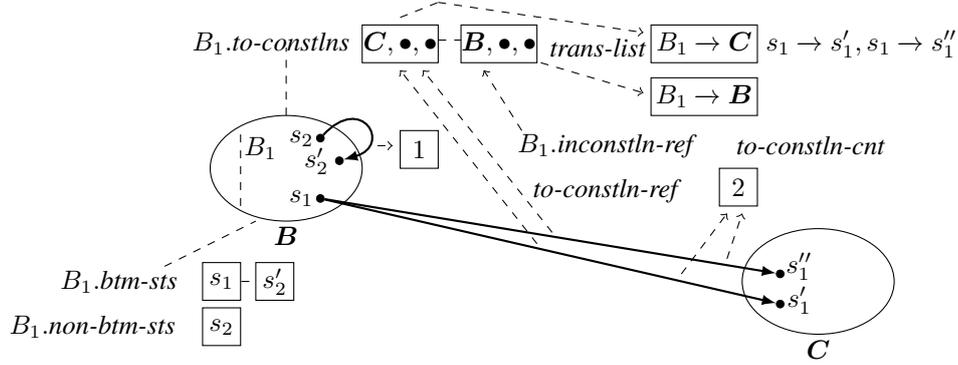
\begin{figure}[h]
\begin{center}
\scalebox{1.0}{
\begin{tikzpicture}
%\draw (0,0.5) ellipse (1cm and 0.5cm);
%\draw (0,-0.3)  node (D) {$D$};
\draw (3,0.5) ellipse (1cm and 0.7cm) node[above left] (B2) {$B_1$};
\draw (3,-0.4) node (B) {$\setB$};
\draw (10,-1.0) ellipse (1cm and 0.7cm);
\draw (10,-1.9) node (C) {$\setC$};
\draw[dashed] (2.4,0.0) --  (2.4,1.0);
%\draw (2.65,0.75) node[circle, minimum size=1mm, fill, draw, inner sep=0pt] (n1) {};
%\draw[fill] (2.65,0.25) node[circle, minimum size=1mm, fill, draw, inner sep=0pt] (n2) {};
\draw[fill] (3.45,0.1) node[circle, minimum size=1mm, fill, draw, inner sep=0pt] (n1) {};
\draw (3.2,0.1) node (s1) {$s_1$};
\draw[fill] (3.45,0.9) node[circle, minimum size=1mm, fill, draw, inner sep=0pt] (n4) {};
\draw (3.2,0.9) node (s2) {$s_2$};
\draw[fill] (3.7,0.6) node[circle, minimum size=1mm, fill, draw, inner sep=0pt] (n5) {};
\draw (3.4,0.6) node (s2) {$s_2'$};
%\draw[fill] (0,0.5) node[circle, minimum size=1mm, fill, draw, inner sep=0pt] (n4) {};
\draw[fill] (9.5,-1.3) node[circle, minimum size=1mm, fill, draw, inner sep=0pt] (n2) {};
\draw (9.75,-1.25) node (s1p) {$s_1'$};
\draw[fill] (9.5,-0.9) node[circle, minimum size=1mm, fill, draw, inner sep=0pt] (n3) {};
\draw (9.75,-0.8) node (s1pp) {$s_1''$};
%\draw[->] (n1) -- (n3);
%\draw[->] (n2) -- (n3);
\draw[thick,-latex] (n1) -- (n2);
\draw[thick,-latex] (n1) -- (n3);
%\draw[thick,->] (n4) [out=270,in=0,loop] (n4);
\draw[thick,-latex] (n4) .. controls (4,1.5) and (4.4,0.8) .. (n5);
%\draw[->] (n1) -- (n4);
%\draw[->] (n2) -- (n4);
\draw (4.0,1.9) rectangle (5.0,2.4);
\draw (4.5,2.15) node (tc1) {$\setC, \bullet, \bullet$};
\draw[dashed] (5.0,2.2) --  (5.3,2.2);
\draw (5.3,1.9) rectangle (6.3,2.4);
\draw (5.8,2.15) node (tc2) {$\setB, \bullet, \bullet$};
\draw[dashed] (3.0,1.2) --  (3.0,2.1);
\draw (2.8,2.15) node (toc) {$B_1.\toconstellations$};

\draw (8.7,0.0) rectangle (9.2,0.5);
\draw (8.95,0.25) node (sc1) {$2$};
\draw (9.9,0.8) node (sctext) {$\toconstellationcount$};
\draw[->,dashed] (8.2,-0.9) -- (8.8,-0.1);
\draw[->,dashed] (8.8,-0.7) -- (9.0,-0.1);

\draw (7.2,0.2) node (sctext) {$\toconstellationref$};

\draw (7.2,0.8) node (intext) {$B_1.\inconstellationref$};
\draw[->,dashed] (6.1,1.0) -- (5.6,1.8);

\draw[->,dashed] (6.3,-0.5) -- (4.5,1.8);
\draw[->,dashed] (6.5,-0.35) -- (4.8,1.8);

\draw (4.5,0.5) rectangle (5.0,1.0);
\draw (4.75,0.75) node (sc2) {$1$};
\draw[->,dashed] (4.2,0.8) -- (4.4,0.8);
%\draw[->,dashed] (4.2,1.1) -- (5.6,1.8);

%\draw[->,dashed] (6.9,0.1) -- (8.3,-0.95);
%\draw[->,dashed] (5.9,0.25) -- (7.7,-0.95);

%\draw (5.6,-1.3) node (sctext) {$\blockconstellationlist$};

\draw (7.8,1.9) rectangle (9.2,2.4);
\draw (8.5,2.15) node (p1) {$B_1 \pijl{}{} \setC$};
\draw (10.57,2.15) node (p1) {$s_1 \pijl{}{} s_1', s_1 \pijl{}{} s_1''$};

%\draw[->,dashed] (4.1,0.6) -- (7.7,-1.7);

\draw (7.8,1.2) rectangle (9.2,1.7);
\draw (8.5,1.45) node (p1) {$B_1 \pijl{}{} \setB$};
%\draw (9.9,1.45) node (p1) {$s_2 \pijl{}{} s_2'$};

\draw[->,dashed] (6.35,1.9) -- (7.7,1.5);
\draw (7.1,2.1) node (sctext) {$\transitionlist$};
\draw[dashed] (4.5,2.5) -- (5.0,2.7);
\draw[->,dashed] (5.0,2.7) -- (7.7,2.3);

\draw (0.8,-1.0) node (sctext) {$B_1.\bottomstates$};
\draw (0.45,-1.6) node (sctext) {$B_1.\nonbottomstates$};

\draw (1.9,-1.25) rectangle (2.4,-0.75);
\draw (2.15,-1.0) node (bs1) {$s_1$};

\draw (2.6,-1.25) rectangle (3.1,-0.75);
\draw (2.85,-1.0) node (bs1) {$s_2'$};

\draw[dashed] (2.41,-1.0) --  (2.61,-1.0);

\draw (1.9,-1.85) rectangle (2.4,-1.35);
\draw (2.15,-1.6) node (nbs1) {$s_2$};

\draw[dashed] (1.4,-0.8) --  (2.6,-0.2);

%\draw[dashed] (7.8,1.25) --  (9,1.8);
\end{tikzpicture}
}
\end{center}
\caption{An example showing some of the data structures used in the detailed algorithm.}
\label{fig:structures}
\end{figure}

Figure~\ref{fig:structures} illustrates some of the used structures. A block $B_1$ in constellation $\setB$ contains bottom states 
$s_1$, $s_2'$ and non-bottom state $s_2$.
%These are listed in $B_1.\bottomstates$ and $B_1.\nonbottomstates$, respectively. Furthermore, 
State $s_1$ has two outgoing transitions $s_1 \pijl{}{} s_1'$, $s_1 \pijl{}{} s_1''$ to states $s_1'$, $s_1''$ in constellation $\setC$. This means that for both transitions, we have the following references:
\begin{enumerate}
\item $\toconstellationcount$ to the number of outgoing transitions from $s_1$ to $\setC$.
\item $\toconstellationref$ to the element $(\setC, \bullet, \bullet)$
in $B_1.\toconstellations$, where $\setC$ is the constellation reached by the transitions, and the $\bullet$'s are the (now uninitialized)
references that are used when splitting.
\item Via element $(\setC, \bullet, \bullet)$, a reference $\transitionlist$ to the list of transitions
from $B_1$ to $\setC$.
\end{enumerate}
Note that for the inert transition $s_2 \pijl{}{} s_2'$, we only have a reference to the number of outgoing transitions from $s_2$ to $\setB$, and that $B_1.\inconstellationref$ refers to the element $(\setB,\bullet,\bullet)$
associated with the constellation containing $B_1$.
\end{enumerate}

\subsection{Finding the blocks that must be split}

While \nontrivialconstellations is not empty, we perform the algorithm listed in the following sections.
To determine whether the current partition $\pi$ is unstable, we select a constellation $\setB$ 
in \nontrivialconstellations, and we select a block $B$ from $\setB.\blocks$ such that $|B|\leq\frac{1}{2}\setB.\size$. 
We first check which blocks are unstable for $\B$ and $\setB \setminus B$.
\begin{enumerate}
\item We update the list of constellations w.r.t.\ $B$ and $\setB \setminus B$.
\begin{enumerate}
\item Move $\B$ into a new constellation \setC (with $\setC.\blocks$ empty and $\setC.\size = 0$), by removing $\B$ 
from $\setB.\blocks$, adding it to $\setC.\blocks$.
\item 
Add \setC to \trivialconstellations. If $|\setB.\blocks| {=}1$, move $\setB$ to $\trivialconstellations$.
%\item
%For each $s \in \B$, $(s \naarpijl{} s') \in s.T_{\mathit{src}}$ ($s' \in B'$), 
%do the following:
%\begin{enumerate}
%\item 
%Consider element $l$ in $B'.\toconstellations$ referred to by $(s \naarpijl{} s').\toconstellationref$.
%If the reference in $l$ to a new element is undefined, create it, call it $l'$, set its constellation to $\setC$, its
%count to $0$ and put it in $B'.\toconstellations$. Otherwise, call the new element $l'$. Decrease the count of $l$ and if it becomes $0$ remove $l$. Set $(s \naarpijl{} s').\toconstellationref$
%to $l'$ and increase its count.
%\end{enumerate}
%At the end, traverse for all states $s {\in} \B$ the incoming transitions 
%$(s \naarpijl{} s') \in s.T_{\mathit{src}}$ again, to reset the temporary
%references.
%% in the lists $B'.\toconstellations$ to undefined.
%%and $|\setB| = 2$, we remove $\setB$ from the list of constellations, and reset $$B'$.\constellation$ for all $$B'$ \in \setB$. Otherwise, if \splittableblocks is not empty and $|\setB| > 2$, only reset $\B.\constellation$.
\end{enumerate}

\item
Walk through the elements in
\B.\bottomstates and \B.\nonbottomstates, i.e., the states in $B$.
For each state $s {\in} \B$ visit all transitions $s \naarpijl{} s' \in s.T_{\mathit{src}}$, 
and do the steps (a) to (d) below for each transition where $B \neq B'$, with $B'$ the block in which $s'$ resides.
%Also perform step (a).i if (a) holds, followed by step (c) with $s'$ taken equal to $s$, and
%$B'$ taken equal to $B$. This is to handle the special case that states in $B$ can be reached by states in $B$, in particular themselves, via zero transitions.
\begin{enumerate}
\item If $B'.\markedbottomstates$ and $B'.\markednonbottomstates$ are empty:

\begin{enumerate}
\item Put $B'$ in a list \splittableblocks.
\item Let $B'.\coconstellationref$ refer to $(s \naarpijl{} s') .\toconstellationref$.
Let $B'.\constellationref$ refer to a new element in $B'.\toconstellations$.
\end{enumerate}
\item 
If $s'.\coconstellationcounter$ is uninitialized, 
let $s'.\constellationcounter$ refer to a new counter (with initial value $0$), and let $s'.\coconstellationcounter$ refer to $(s \naarpijl{} s').\toconstellationcount$.
%copy the value of 
%$(s \naarpijl{} s').\toconstellationcount$ to a new variable and let $s'.\coconstellationcounter$ refer to it.
%Set $(s \naarpijl{} s').\toconstellationcount$ to zero and let $s'.\constellationcounter$ refer to this variable.
\item 
If $s' \not\in B'.\markedbottomstates$ and $s' \not\in B'.\markednonbottomstates$, then:
\begin{enumerate}
\item If $s'$ is a bottom state, 
move $s'$ from $B'.\bottomstates$ to $B'$.\markedbottomstates.
\item Else move $s'$ from $B'$.\nonbottomstates to $B'$.\markednonbottomstates.
\end{enumerate}
\item Increment the variable to which $s'.\constellationcounter$ refers, and
 let $(s\naarpijl{} s').\toconstellationcount$ refer to this variable. 
Decrement the variable referred to by 
$s'.\coconstellationcounter$.
Move $s \naarpijl{} s'$ from $B'.$$\coconstellationref.\transitionlist$
to $B'.\constellationref.\transitionlist$, and let $(s\naarpijl{}{}s').\toconstellationref$ refer to $B'.\constellationref$.
%\item
%If $B'.\toconstellationelement$ is not yet initialised, initialise it to $(s\naarpijl{}s').\toconstellationref$.
\end{enumerate}
\item Next, check whether $B$ itself can be split. First, mark all states by moving the states in $B.\bottomstates$ to $B.\markedbottomstates$
and those in $B.\nonbottomstates$ to $B.\markednonbottomstates$. Add $B$ to $\splittableblocks$ and reset references $B.\constellationref$
and $B.\coconstellationref$. Next, for each $s{\in}B$ visit all transitions $s \pijl{}{} s' \in s.T_{\mathit{tgt}}$, and do the steps (a) to (c) below for each transition
where either $\setBp = \setB$ or $\setBp = \setC$, with $\setBp$ the constellation in which $s'$ resides.
\begin{enumerate}
\item If $\setBp {=} \setB$ and $B.\constellationref$ is uninitialized,
let $B.\coconstellationref$ refer to $(s \pijl{}{} s').\toconstellationref$,
add a new element for $\setC$ to $B.\toconstellations$, and let $B.\constellationref$ and
$B.\inconstellationref$ refer to this element.
\item If $s.\coconstellationcounter$ is uninitialized, let $s.\constellationcounter$ refer to a new counter, and 
$s.\coconstellationcounter$ to $(s \pijl{}{} s').\toconstellationcount$.
\item If $B' = B$, increment the variable to which $s.\constellationcounter$ refers, let $(s \pijl{}{} s').\toconstellationcount$ refer
to this variable, and decrement the variable referred to by $s.\coconstellationcounter$.
\end{enumerate}
\item Do the steps below for each $B'{\in} \splittableblocks$.
%Note that in this case there is at least one transition from $B'$ to $B$.
\begin{enumerate}
\item If $|B'.\bottomstates|>0$ (there is at least one unmarked bottom state in $B'$),
the block must be split. 
We leave $B'$ in the list of $\splittableblocks$.
\item Else, if $B'.\coconstellationref.\transitionlist$ is not empty, and there 
is a state $s {\in} B'.\markedbottomstates$ with $s.\coconstellationcounter$ uninitialized or $0$, the block must be split.
\item Else, no splitting is required. 
%If $|B'.\bottomstates|=0$ replace the constellation in $B'.\toconstellationelement$ by the constellation $\setC$. 
Remove $B'$ from \splittableblocks and \textbf{remove the temporary markings} of $B'$ by doing steps i to iii below.
\begin{enumerate}
\item Move each $s{\in}B'.\markedbottomstates$ to $B'.\bottomstates$ and reset $s.\constellationcounter$ and $s.\coconstellationcounter$.
\item Move each $s{\in}B'.\markednonbottomstates$ to $B'.\nonbottomstates$. Reset $s.\constellationcounter$.\\ If $s.\coconstellationcounter = 0$,
delete the variable to which $s.\coconstellationcounter$ refers.\\ Reset $s.\coconstellationcounter$.
\item Do the following steps for $\reff = \constellationref$ and $\reff = \coconstellationref$, if $B'.\reff$ is initialized.
\begin{enumerate}
\item If $|B'.\reff.\transitionlist| = 0$, then first reset $B'.\inconstellationref$ if the element to which $B'.\reff$ refers is associated with $B'.\constellation$,
and second remove the element to which $B'.\reff$ refers from $B'.\toconstellations$ and delete it.
\item Reset $B'.\reff$.
\end{enumerate} 
%$B'.\markedbottomstates$ 
%(resp.\ $B'.\markednonbottomstates$) must be moved to $B'.\bottomstates$ (resp.\ $B'.\nonbottomstates$).
%If $B'.\constellationtransitions$ or $B'.\coconstellationtransitions$ is equal to the empty list,
%remove this empty list.

\end{enumerate}\end{enumerate}

%Otherwise, if $\B.\constellationtransitioncounter$
%is larger than zero and 
%the $\mathit{constellation\_bottom\_mark\_counter}$ is less than 
%$\mathit{number\_of\_bottom\_states}$ there are states with transitions to $\unionsetB\setminus B$
%without transitions to \B, and the block must also been split.

\item If $\splittableblocks$ is not empty, start splitting (section~\ref{sec:splitting}). 
Else, carry on with finding blocks to split, by selecting another non trivial constellation $\setB$ and block 
$B{\in} \setB$, and continuing with step 1. If there are no non trivial constellations left, 
the current partition is stable, the algorithm terminates.

\end{enumerate}

\subsection{Splitting the blocks}
\label{sec:splitting}

Splitting the splittable blocks is performed using the following steps.
We walk through the blocks $B'$ in \splittableblocks, which must be split into two or three blocks under constellation \setB and block $\B$.
For each splitting procedure, we use time proportional to the smallest of the two blocks into which $B'$ is split, where one of the two smallest
blocks can have size $0$, and we also allow ourselves to traverse the marked states, which is proportional to the time
we used to mark the states in the previous step. Create new lists $X_{B'}$, $X_{B''}$, $X_{B'''}$ to keep track of new bottom states
when splitting.

If $|B'.\bottomstates|=0$ (all bottom states are marked), 
then we have $\splitpi(\B', \B ) = \B'$, and can start with step 3 below.
\begin{enumerate}
\item
We start to split block $B'$ w.r.t.\ $B$. We must determine whether $\splitpi(B', \B)$ or
$\cosplitpi(B', \B)$ is the smallest. This is done by performing the following two
procedures in lockstep, alternatingly processing a transition.
% in one of the procedures.
The entire operation terminates when one of the procedures terminates.
If one procedure acquires more states than $\frac{1}{2}|B'|$, it is stopped,
and the other is allowed to terminate.
\begin{enumerate}
\item The first procedure attempts to collect the states in $\splitpi(B',\B)$.
\begin{enumerate}
\item Initialise an empty stack $Q$ and an empty list $L$. Let $D_1$ refer to the list consisting of $B'.\markedbottomstates$ and $B'.\markednonbottomstates$.
In the next step, we walk through the states in $D_1$.
%\item Push all the states in $B'.\markedbottomstates$ and $B'.\markednonbottomstates$ on $Q$ and add them to $L$.
\item Perform \textbf{detect1} as long as $|L| \leq \frac{1}{2}|B'|$:
\begin{enumerate}
\item While $Q$ is not empty or we have not walked through all states in $D_1$, do the following steps.
\begin{itemize}
\item If $Q$ is empty, push the next state in $D_1$ on $Q$ and add it to $L$.
\item 
Pop $s$ from $Q$. For all $s \naarpijl{}s' \in s.T_{\mathit{src}}$
if $s' \in B'$ and $s' \not\in L$, add $s'$ to $L$ and
push $s'$ on $Q$.
\end{itemize}
\end{enumerate}
\end{enumerate}
\item
The second procedure attempts to collect the states in $\cosplitpi(B',B)$. It uses a priority queue $P$, in which the priority 
of a state $s$ represents the number of outgoing inert transitions to a target state 
for which it has not yet been determined whether it is in $\cosplitpi(B',\B)$. If the priority of a state $s$ 
becomes $0$, it is obvious that $s$ must be in $\cosplitpi(B',\B)$. 
\begin{enumerate}
\item Create an empty priority queue $P$ and an empty list $L'$. Let $D_2$ refer to the list $B'.\bottomstates$.
%\item Add the remaining states in $B'.\bottomstates$ with priority $0$ to $P$.
\item Perform \textbf{detect2} as long as $|L'| \leq \frac{1}{2}|B'|$:
\begin{enumerate}
\item While $P$ has states with priority $0$ or we have not walked through all states in $D_2$, 
do the following steps.
\begin{itemize}
\item If we have not walked through all states in $D_2$, let $s$ be the next state in $D_2$. Else, get a state with priority $0$ from $P$,
and let $s$ be that state. Add $s$ to $L'$. 
%perform 
\item For all $s \naarpijl{} s' \in s.T_{\mathit{src}}$, do the following steps.
\begin{itemize}
\item If $s' {\in} B'$, $s' {\not\in} P \cup L'$, and $s' {\not\in} B'.\markednonbottomstates$ (or $s'$ does not have a transition to $\setB\setminus B$; this last condition is required when \textbf{detect2} is invoked in 5.3.4.b and 5.3.7.b.i.B, and can be checked for $s'{\in}\markednonbottomstates$ by determining whether the variable to which $s'.\coconstellationcounter$ refers,
minus $s'.\inertcounter$ if $\setBp = \setB$, is larger than $0$. Else, it can be checked by walking over the transitions $s' \pijl{}{} s'' \in s.T_{\mathit{tgt}}$),
add $s'$ with priority $s'.\inertcounter$ to $P$.
\item If $s'{\in} P$, decrement the priority of $s'$ in $P$.
\end{itemize}
\end{itemize}
\end{enumerate}
\end{enumerate}

%For this purpose the remaining states in $\mathit{bottom\_states}$ , which are states without an $a$-transition
%to \B, are put into a list for the new block. Each such state
%$s$ is counted in the variable $\mathit{cosplit\_count}$. If $\mathit{cosplit\_count}$ exceeds 
%$\frac{1}{2} \cdot \mathit{number\_of\_states}$ this process is aborted.
%Furthermore, for each such as state $s$
%the incoming transitions $s'\pijl{\tau}{}s$ are traversed. If $s'$ is marked and if $s'$ is not yet put in a priority queue 
%$\mathit{priority\_queue}$ containing the states to be processed based on the value of $\mathit{inert\_count}$. 
%
%The priority in the queue is decreased by one independently on whether $s'$ is just put in the queue or not.
%
%For each state in the priority queue with priority $0$ the same steps as above for states in $\mathit{bottom\_states}$
%are repeated. 
%
%If at some point there are no states with priority $0$, all states that will not be marked have been obtained
%and the process stops. 
\end{enumerate}
\item
The next step is to actually carry out the splitting of $B'$.
Create a new block $B''$ with empty lists, and add it to the list
of blocks. Set $B''.\constellation$ to $B'.\constellation$, and add $B''$ to the list of blocks of that constellation.

Depending on whether \textbf{detect1} or \textbf{detect2} terminated in the previous step, one of the lists
$L$ or $L'$ contains the states to be moved to $B''$. Below we refer to this list as $N$.
For each $s {\in} N$, do the following:
\begin{enumerate}
\item
Set $s.\block$ to $B''$, and move $s$ from the list in which it resides in $B'$ to the corresponding list in $B''$.
%occurs, move $s$ from $B'.\bottomstates$ to $B''.\bottomstates$, from list
%$B'.\markedbottomstates$ to $B''.\markedbottomstates$, from $B'.\nonbottomstates$ to $B''.\nonbottomstates$,
%and from $B'.\markednonbottomstates$ to $B''.\markednonbottomstates$.
\item
For each $s \pijl{}{} s' {\in} T_{\mathit{tgt}}$, do the following steps.
\begin{enumerate}
\item If $(s \pijl{}{} s').\toconstellationref$ is initialized, i.e., $s \pijl{}{} s'$ is not inert, 
consider the list element $l$ in $B'.\toconstellations$ retrievable by $(s\pijl{}{}s').\toconstellationref$.
Check whether there is a corresponding new element in $B''.\toconstellations$. If so, $l$ refers to this new
element, which we call $l'$.
If not, create it, add it to $B''.\toconstellations$ and call it also $l'$, set the constellation in this new $l'$ to that of $l$, 
set $B''.\inconstellationref$ to $l'$ in case this constellation is $B''.\constellation$,
set $B''.\constellationref$ to $l'$ in case $l$ refers to $B'.\constellationref$,
set $B''.\coconstellationref$ to $l'$ in case $l$ refers to $B'.\coconstellationref$,
and let $l$ and $l'$ refer to each other.
Move $s \pijl{}{} s'$ from $l.\transitionlist$ to $l'.\transitionlist$ and let $(s \pijl{}{} s').\toconstellationref$ refer to $l'$.
\item Else, if $s'{\in} B' \setminus N$ (an inert transition becomes non-inert):
\begin{enumerate}
\item Decrement $s.\inertcounter$.
\item If $s.\inertcounter {=} 0$, add $s$ to $X_{B''}$, move $s$ from $B''.\nonbottomstates$ or
$B''.\markednonbottomstates$ to the corresponding bottom states list in $B''$.
If $B''.\inconstellationref$ is uninitialized, create a new element for $B''.\constellation$, add it to
$B''.\toconstellations$, let $B''.\inconstellationref$ refer to that element, and if $B'.\inconstellationref$ is initialized, let $B'.\inconstellationref$
and $B''.\inconstellationref$ refer to each other. Add $s \pijl{}{} s'$ to $B''.\inconstellationref.$ $\transitionlist$ and let $(s \pijl{}{} s').\toconstellationref$
refer to $B''.\inconstellationref$.
\end{enumerate}
\end{enumerate}
\item For each $s \naarpijl{} s' \in T_{\mathit{src}}$, $s' \in B' \setminus N$ (an inert transition becomes non-inert):
\begin{enumerate}
\item Decrement $s'.\inertcounter$.
\item If $s'.\inertcounter {=} 0$, add $s'$ to $X_{B'}$, and move $s'$ from $B'.\nonbottomstates$~or
$B'.\markednonbottomstates$ to the corresponding bottom states list in $B'$.
If $B'.$ $\inconstellationref$ is uninitialized, create a new element for constellation $B'.\constellation$ and add it to
$B'.\toconstellations$, let $B'.\inconstellationref$ refer to that element, and if $B''.\inconstellationref$ is initialized, let $B'.\inconstellationref$
and $B''.$ $\inconstellationref$ refer to each other. Add $s \naarpijl{}{} s'$ to $B'.\inconstellationref.$ $\transitionlist$ and let $(s \naarpijl{}{} s').$ $\toconstellationref$
refer to $B'.\inconstellationref$.
\end{enumerate}

%either from list $B'.\constellationtransitions$ to $B''.\constellationtransitions$, or from list $B'.\coconstellationtransitions$ to $B''.\coconstellationtransitions$. 

% Let $(s\pijl{}{}s').\blockconstellationlist$ refer to the new list the transition is in.
%\item 
%For each $s \pijl{}{} s' \in T_{\mathit{tgt}}$, do the following:
%\begin{enumerate}
%\item
%If $s.\block=B'$ skip this step. 
%If $s.\block=B''$, 
%\end{enumerate}
\end{enumerate}
\item For each element $l$ in $B''.\toconstellations$ referring to an element $l'$ in $B'.\toconstellations$, do the following steps.
\begin{enumerate}
\item If $l'.\transitionlist$ is empty and $l'$ does not refer to $B'.\constellationref$ and not to $B'.$ $\coconstellationref$, reset $B'.\inconstellationref$ if $l'$ is associated with $B'.\constellation$, remove $l'$ from $B'.\toconstellations$ and delete it.
\item Else, reset the reference from $l'$ to $l$.
\item Reset the reference from $l$ to $l'$.
\end{enumerate}
\item Next, we must consider splitting $\splitpi(B', B)$ under $\setB \setminus B$, 
or if $B'$ was stable under $B$, we must split $B'$ under $\setB \setminus B$. 
Define $C=\splitpi(B', B)$. If $B'$ was not split, then $C=B'$.
$C$ is stable under $\setB \setminus B$ if $C.\coconstellationref$ is uninitialized or $C.\coconstellationref.\transitionlist$ is empty
or for all $s{\in} C.\markedbottomstates$ it holds that $s.\coconstellationcounter>0$.
%Assume that this is not the case, then there are two rather different cases.
%\begin{enumerate}
%\item Assume $\splitpi(B', B)$ is the smaller set (where $B'$ has been split). 
%Iterate over all states in $\splitpi(B', B)$ and split the set in two parts. As all states
%and transitions to and from $\splitpi(B', B)$ can be visited without increasing the complexity of the algorithm, this is a straightforward operation.
%\item Assume $\cosplitpi(B', B)$ is the smaller set or $B'$ has not been split. 
%We must 
If this is not the case, then we must determine which of the blocks
$\splitpi(C, \setB \setminus B)$ or 
$\cosplitpi(C, \setB \setminus B)$ is the smallest in a time proportional to the smallest of the two. 
This is again done by simultaneously iterating over the
transitions of states in both sets in lockstep.
The entire operation terminates when one of the two procedures terminates.
If one of the procedures acquires more than $\frac{1}{2}|C|$ states, that
procedure is stopped, and the other is allowed to terminate.
%Before starting the two procedures, we set both $\splitcounter$ and $\cosplitcounter$ to $0$,
%and empty $\splitstatesfornewblock$, stack $Q$, and priority queue $P$. Furthermore, for all states $s$ in $B'.\markedbottomstates$ and $B'.\markednonbottomstates$,
%we set $s.\nestedmarked$ to $\TRUE$.

\begin{enumerate}
\item The first procedure attempts to collect the states in $\splitpi(C, \setB \setminus B)$.
\begin{enumerate}
\item Create an empty stack $Q$ and an empty list $L$. Let $D_1$ be the list of states
$s$ occurring in some $s\pijl{}{}s'$ in the list $\splitpi(B', B).$$\coconstellationref.\transitionlist$ (in practice we walk over $D_1$ by walking
over the latter list).
\item Perform \textbf{detect1} with $B' = C$.
\end{enumerate}

\item The second procedure attempts to collect the states in $\cosplitpi(C, \setB \setminus B)$.
\begin{enumerate}
\item Create an empty priority queue $P$ and an empty list $L'$. Let $D_2$ be the list of states $s$ with
$s.\coconstellationcounter = 0$ in $C.\markedbottomstates$ (in practice we walk over the latter list and check the condition).
\item Perform \textbf{detect2} with $B' = C$.
\end{enumerate}

\end{enumerate}
Finally, we split $C$ by moving either $\splitpi(C, \setB \setminus B)$ or
$\cosplitpi(C, \setB \setminus B)$ to a new block $B'''$, depending on which of the two is the smallest. 
This can be done by using the procedure described in steps 2 and 3, where in step 2.b.ii.B, we fill a state list $X_{B'''}$ instead of $X_{B''}$, and in 
2.c.ii, we possibly add states to $X_C$ instead of $X_{B'}$ (we define $X_C = X_{B'}$ if $C = B'$, otherwise $X_C = X_{B''}$). 
We move all states in $X_C$ that are now in $B'''$ to $X_{B'''}$.

\item Remove the temporary markings of each block $C$ resulting from the splitting of $B'$ (see steps 5.2.4.c.i to iii).

% removed: for each $B'$ that is investigated
%\end{enumerate}

\item If the splitting of $B'$ resulted in new bottom states (either $X_{B'}$, $X_{B''}$, or $X_{B'''}$ is not empty), check for those states whether
further splitting is required. This is the case if from some new bottom states, not all constellations can be reached which can be reached
from the block. Perform the following for $\hat{B} = B'$, $B''$, and $B'''$, and all $s \in X_{\hat{B}}$:
\begin{enumerate}
\item For all $s \pijl{}{} s' \in T_{\mathit{tgt}}$ ($s' \in \setBp$):
\begin{enumerate}
\item If it does not exist, create an empty list $S_\setBp$ and associate it with $\setBp$ in $\hat{B}.\toconstellations$ (accessible via $(s \pijl{}{} s').\toconstellationref$), and move
$\setBp$ to the front of $\hat{B}.\toconstellations$.
\item If $s \not\in S_\setBp$, add it.
\end{enumerate}
\item Move $s$ from $X_{\hat{B}}$ to $\hat{B}.\newbottomstates$.
\end{enumerate}

\item Check if there are unstable blocks that require further splitting, and if so, split further. Repeat this until no further splitting is required. A stack $Q'$ is used to keep track of the blocks that require checking. For $\hat{B}=B'$, $B''$, and $B'''$, push $\hat{B}$ on $Q'$ if $|\hat{B}.\newbottomstates| > 0$. While $Q'$ is not empty, perform steps a to c below.
\begin{enumerate}
\item Pop block $\hat{B}$ from $Q'$.
\item Find a constellation under which $\hat{B}$ is not stable by walking through the $\setB \in \hat{B}.\toconstellations$. If $|S_\setB| < |\hat{B}.\newbottomstates|$,
%, then not all new bottom states can reach $\setB$
 then further splitting is required under $\setB$:
\begin{enumerate}
\item Find the smallest subblock of $\hat{B}$ by performing the following two procedures in lockstep.
\begin{enumerate}
\item The first procedure attempts to collect the states in $\splitpi(\hat{B}, \setB)$.
\begin{itemize}
\item Create an empty stack $Q$ and an empty list $L$. Let $D_1$ be the list of states
$s$ occurring in some $s\pijl{}{}s'$ with $s'{\in}\setB$ in the list $\transitionlist$ associated with $\setB{\in} \hat{B}.\toconstellations$ (in practice we walk over $D_1$ by walking
over $\transitionlist$).
\item Perform \textbf{detect1} with $B' = \hat{B}$.
\end{itemize}

\item The second procedure attempts to collect the states in $\cosplitpi(\hat{B}, \setB)$.
\begin{itemize}
\item Create an empty priority queue $P$ and an empty list $L'$. Let $D_2$ be the list of states $s{\in} \hat{B}.\newbottomstates \setminus S_\setB$.
\item Perform \textbf{detect2} with $B' = \hat{B}$.
\end{itemize}
\end{enumerate}
\item Continue the splitting of $\hat{B}$ by performing step 2 to produce a new block $\hat{B}'$ and a list of new bottom states $X_{\hat{B}'}$.
Move all states in $\hat{B}.\newbottomstates$ that have moved to $\hat{B}'$ to $\hat{B}'.\newbottomstates$. Update the $S_\setB$ lists by walking over the $l {\in} \hat{B}.\toconstellations$ and doing the following steps as long as an empty $S_\setB$ is not encountered. Note that the $l$ still
refer to corresponding elements $l'$ in $\hat{B}'.\toconstellations$.
\begin{enumerate}
\item For all $s {\in} S_\setB$, if $s {\in} \hat{B}'$, then insert $s$ in the $S_\setB$ associated with $l'$ (if this list does not exist yet, create it, and move $l'$ to the front of $\hat{B}'.\toconstellations$) and remove $s$
from the $S_\setB$ associated with $l$.
\item If the $S_\setB$ of $l$ is now empty, remove it, and move $l$ to the back of $\hat{B}.\toconstellations$.
\end{enumerate}

\item Perform step 3 for elements $l'$ in $\hat{B}'.\toconstellations$ referring to an element $l$ in $\hat{B}.\toconstellations$. Skip steps 4 and 5,
and continue with step 6 for $\hat{B}$, $\hat{B}'$, $X_{\hat{B}}$ and $X_{\hat{B}'}$.
\item Push $\hat{B}$ on $Q'$ if $|\hat{B}.\newbottomstates| > 0$ and $\hat{B}'$ on $Q'$ if $|\hat{B}'.\newbottomstates| > 0$.
%\begin{enumerate} 
%\item If \textbf{detect1} terminated, set list $\hat{B}'.\newbottomstates$ to $S_\setB$, else to $\hat{B}.\newbottomstates \setminus S_\setB$.
%\item If \textbf{detect1} terminated, set list $\hat{B}.\newbottomstates$ to $\hat{B}.\newbottomstates \setminus S_\setB$, else to $S_\setB$.
%\item For all $\setB \in \hat{B}.\toconstellations$, walk through the states in $S_\setB$, and if they are in $\hat{B}'.\newbottomstates$, move
%them to a set $S_\setB$ associated
%with a $\setB$ entry in $\hat{B}'.\toconstellations$ (if such an entry does not exist, create it, and let $\setB$ in $\hat {B}.\toconstellations$ refer to it).
%\end{enumerate}
\end{enumerate}
\item If no further splitting was required for $\hat{B}$, empty $\hat{B}.\newbottomstates$ and remove the remaining $S_\setB$ associated with constellations $\setB {\in} \hat{B}.\toconstellations$. 
\end{enumerate}

\item If $B'.\constellation {\in} \trivialconstellations$, move it to \nontrivialconstellations. 

\end{enumerate}

\section{Application to branching bisimulation}
\label{sec:bbisim}
We show that the algorithm can also be used to determine branching bisimulation,
using the transformation 
from \cite{NV95,RSW14}, with complexity $O(m(\log|\mathit{Act}|+\log n))$.
Branching bisimulation is typically applied to labelled transition systems (LTSs). 
%An LTS is a
%three tuple $A = (S, {\it Act}, \pijl{}{})$, with $S$ a finite set of states, ${\it Act}$ a finite set of
%actions including the internal action $\tau$, and $\pijl{}{} \subseteq S \times {\it Act} \times S$ a transition
%relation.

\begin{definition}{Labeled transition system}
\label{dn:lts}
A {\it labeled transition system} ({LTS}) is a three tuple
$A=(S,{\it Act},\pijl{}{})$ where
\begin{enumerate}
\item
$S$ is a finite set of {\it states}. The number of states is generally denoted by $n$.
\item
${\it Act}$ is a finite set of actions including the \textit{internal action} $\tau$.
\item
$\pijl{}{}\,\subseteq S\times {\it Act}\times S$ is
a {\it transition relation}. The number of transitions is generally denoted as by $m$.
\end{enumerate}
It is common to write $t\pijl{a}{}t'$ for $(t,a,t')\in{\pijl{}{}}$.
\end{definition}

There are various, but equivalent, ways to define branching bisimulation. We use the definition below.

%\begin{definition}
%\label{def:branching}
%Consider the LTS $A=(S,{\it Act},\pijl{}{})$.
%We call a symmetric relation $R \subseteq S\times S$ a {\em branching bisimulation relation} iff
%%\vspace{-0.3cm}
%\[
%\begin{aligned}
%& \forall s, t, s' {\in} S. \forall a {\in} {\it Act}. s R t \wedge s \pijl{a}{} s' {\implies} \\
%& (a {=} \tau \wedge s' R t) \vee (\exists t', t'' {\in} S. t \Rightarrow t' \pijl{a}{} t'' \wedge s R t' \wedge s' R t''),
%\end{aligned}
%\]
%%\end{shrinkeq}
%where $\Rightarrow$ is the transitive, reflexive closure of $\pijl{\tau}{}$.
%\end{definition}
%States are {\it branching bisimilar} iff there is a branching bisimulation relation $R$
%relating them. 

\begin{definition}{Branching bisimulation}
\label{def:branching}
Consider the labeled transition system $A=(S,{\it Act},\pijl{}{})$.
We call a symmetric relation $R\subseteq S\times S$
a {\em branching bisimulation relation} iff for all $s,t\in S$ such that $\R{s}{t}$,
the following conditions hold for all actions $a\in{\it Act}$:
\begin{enumerate}
\item
If $s\pijl{a}{}s'$, then
\begin{enumerate}
\item
Either $a=\tau$ and $\R{s'}{t}$, or
\item
There is a sequence
$t\pijl{\tau}{}\cdots\pijl{\tau}{} t'$
of (zero or more) $\tau$-transitions such that $\R{s}{t'}$
and $t'\pijl{a}{}t''$ with $\R{s'}{t''}$.
\end{enumerate}
\end{enumerate}
Two states $s$ and $t$ are {\it branching bisimilar}
%, denoted by $s\bbis t$,
iff there is a branching bisimulation relation $R$ such that
$\R{s}{t}$.
%Two labeled transition systems are {\it branching bisimilar}
%iff their initial states are branching bisimilar.
\end{definition}

Our new algorithm can be applied to an LTS by translating it to a Kripke structure.
\begin{definition}{LTS embedding}
Let $A=(S,{\it Act},\pijl{}{})$ be an LTS. We construct
the \textit{embedding of} $A$ to be the Kripke structure $K_A=(S_A, \mathit{AP},\pijl{}{},L)$ as follows:
\begin{compactenum}
\item
$S_A=S\cup \{\langle a,t\rangle \mid s\pijl{a}{}t$ for some $t\in S \}$.
\item
$\mathit{AP}=\mathit{Act}\cup\{\bot\}$.
\item
$\pijl{}{}$ is the least relation satisfying ($s,t\in S$, $a\in\mathit{Act}\setminus\tau$):
\[\frac{s\pijl{a}{}t}{s\pijl{}{}\langle a,t\rangle} \hspace*{1cm} \frac{}{\langle a,t\rangle\pijl{}{}t} \hspace*{1cm}
\frac{s\pijl{\tau}{}t}{s\pijl{}{}t}\]
\item
$L(s)=\{\bot\}$ for $s\in S$ and $L(\langle a,t\rangle)=\{a\}$.
\end{compactenum}
\end{definition}
The following theorem stems from \cite{NV95}.
\begin{theorem}
Let $A$ be an LTS and $K_A$ its embedding.
%Consider two states $s,t$ in $A$.
Then two states are branching
bisimilar in $A$ iff they are divergence-blind stuttering equivalent in $K_A$.
\end{theorem}
If we start out with an LTS with $n$ states and $m$ transitions then its
embedding has at most $m+n$ states and $2m$ transitions. Hence,
the algorithm requires $O(m\log (n{+}m))$ time. As $m$ is at most $|\mathit{Act}|n^2$ this is
also equal to $O(m(\log|\mathit{Act}|{+}\log n))$.

A final note is that the algorithm can also easily be adapted to determine divergence-sensitive
branching bisimulation~\cite{GW96}, by simply adding a self loop indicating divergence to 
those states on a $\tau$-loop, similar to the way stuttering equivalence is calculated using
divergence-blind stuttering equivalence.

\section{Experiments}

\begin{figure}[t]
\centering
\subfloat{
\centering
\scalebox{0.24}{
\includegraphics{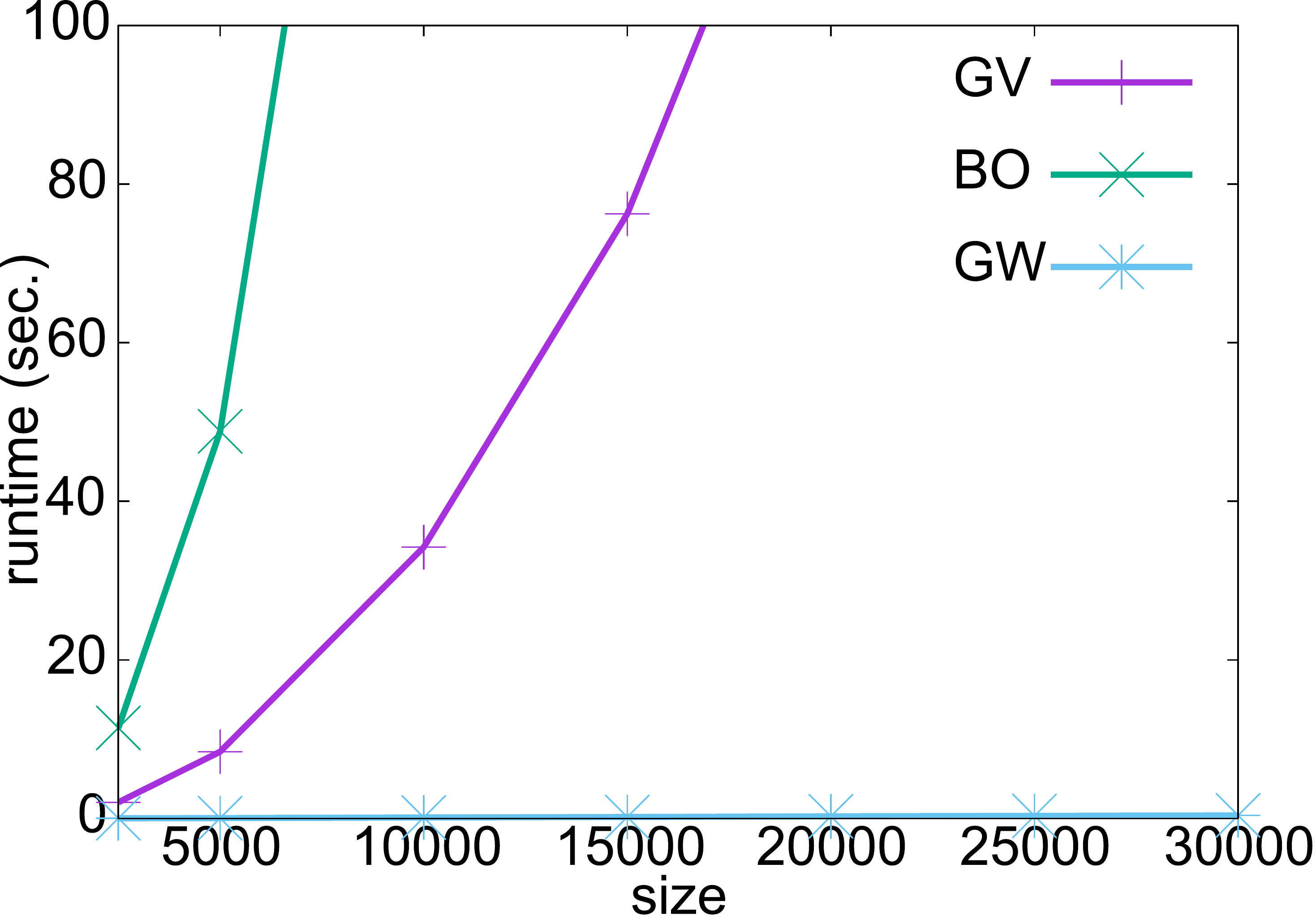}
}
}
\subfloat{
\centering
\scalebox{0.24}{
\includegraphics{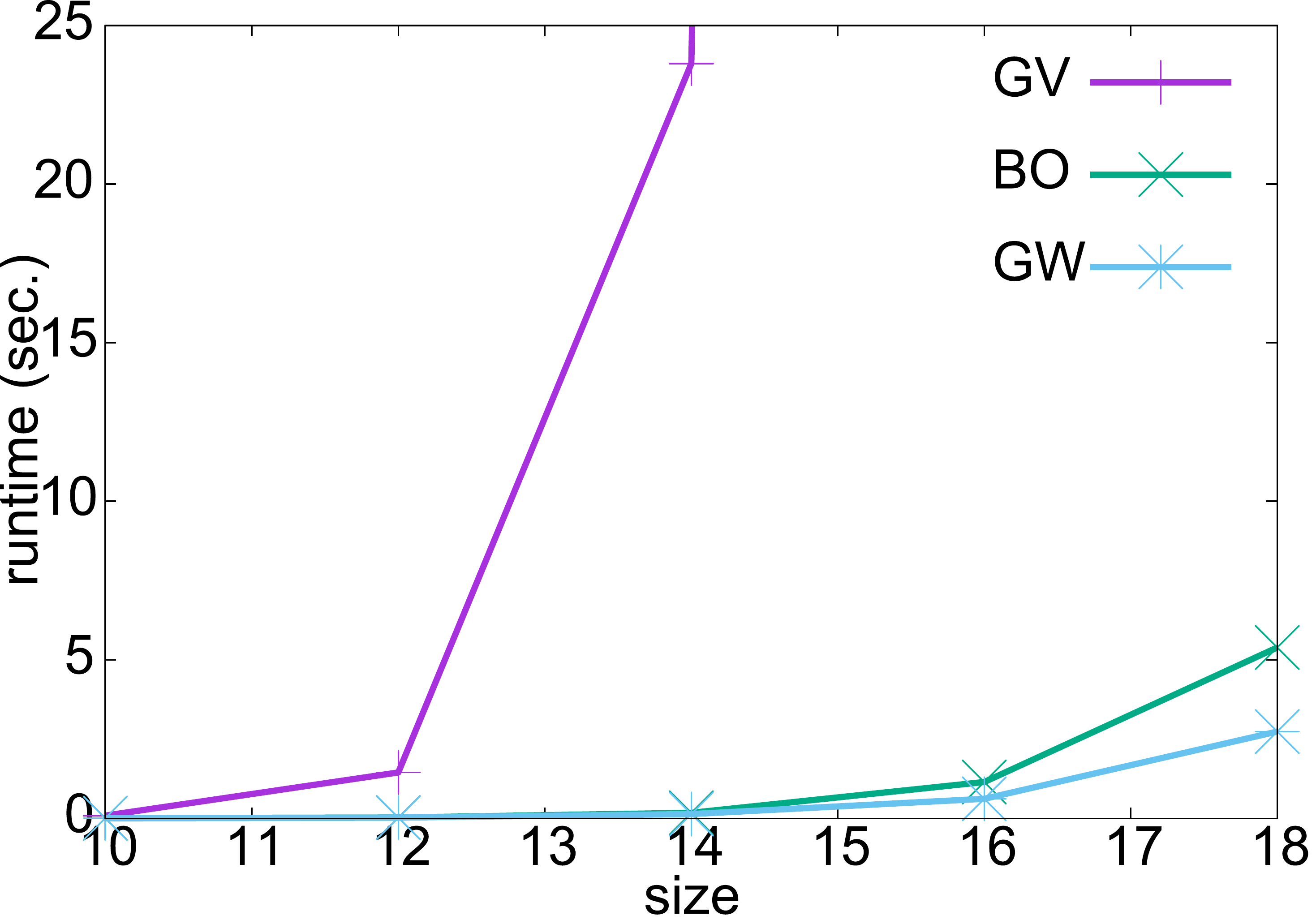}
}
}
\caption{Runtime results for $(a{\cdot} \tau)^\mathit{size}$ sequences (left) and trees of depth $\mathit{size}$ (right)}
\label{fig:results}
\end{figure}

The new algorithm has been implemented as part of the mCRL2 toolset~\cite{CGKSVWW13}, 
which offers implementations of GV and the algorithm
by Blom \& Orzan~\cite{BO05} that distinguishes states by their connection to blocks via their
outgoing transitions. We refer to the latter as BO. The performance of GV and BO can be very different on 
concrete examples.
We have extensively tested the new algorithm by applying it to thousands of randomly
generated LTSs and comparing the results with those of the other algorithms.

We experimentally compared the performance of GV, BO, and the implementation of the new algorithm (GW). All experiments involve the
analysis of LTSs, which for GW are first transformed to Kripke structures using the translation of
section~\ref{sec:bbisim}. The reported runtimes do not include the time to read the input LTS
and write the output, but the time it takes to translate the LTS to a Kripke structure and 
to reduce strongly connected components
is included.

Practically all experiments have been performed on machines running \textsc{CentOS Linux},
with an \textsc{Intel} E5-2620 2.0 GHz CPU and 64 GB RAM. Exceptions to this are the final two entries
in table~\ref{tab:results}, which were obtained by using a machine running \textsc{Fedora} 12, with
an \textsc{Intel Xeon} E5520 2.27 GHz CPU and 1 TB RAM.

Figure~\ref{fig:results} presents the runtime results for two sets of experiments designed to demonstrate
that GW has the expected scalability. At the left are the results of analysing single sequences of the shape 
$(a{\cdot} \tau)^n$. As the length $2n$ of such a sequence is increased,
the results show that the runtimes of both BO and GV increase at least quadratically, while
the runtime of GW grows linearly. All algorithms require $n$ iterations, in which BO and GV walk over all the states in
the sequence, but GW only moves two states into a new block.
At the right of figure~\ref{fig:results}, the results are displayed of analysing trees of depth $n$ that up to level $n{-}1$ 
correspond with a binary tree of $\tau$-transitions. Each state at level $n{-}1$ has a uniquely labelled outgoing transition 
to a state in level $n$. This example is particularly suitable for BO which only needs one iteration to obtain the stable
partition. Still GW beats BO by repeatedly splitting off small blocks of size $2(k-1)$ if a state at level $k$ is the splitter.

Table~\ref{tab:results} contains results for minimising LTSs from the VLTS benchmark set\footnote{\url{http://cadp.inria.fr/resources/vlts}.} and the mCRL2 toolset\footnote{\url{http://www.mcrl2.org}.}.
For each case, the best runtime result has been highlighted in bold.
Some characteristics of each case are given on the left, in particular the number of states ($n$) and transitions ($m$)
in the original LTS and the number of states (\textit{min.} $n$) and transitions (\textit{min.} $m$) in the minimized LTS.

The final two cases stem from mCRL2 models distributed with the mCRL2 toolset as follows:

\begin{itemize}
\item \textbf{dining\_14} is an extension of the Dining Philosophers model to fourteen philosophers;
\item \textbf{1394-fin3} extends the 1394-fin model to three processes and two data elements.
\end{itemize}

The experiments demonstrate that also when applied to actual 
state spaces of real models, GW generally outperforms the best of the other algorithms, often with a factor 10 and sometimes with a factor 100. This difference 
tends to grow as the LTSs
get larger. GW's memory usage is only sometimes substantially higher than GV's and BO's, which surprised us given the amount of 
required bookkeeping.

\begin{landscape}
\begin{table}[pth]
%\scriptsize
   \centering
  \begin{tabular}{|l|r|r|r|r|r|r|r|r|r|r|}
    \hline
    \textbf{Model} & \textbf{  \textit{n}} & \centering \textbf{ \textit{m}} & \centering \textbf{ min.\ \textit{n}} & \centering \textbf{ min.\ \textit{m}} & \textbf{time GV} & \textbf{me.\ GV} & \textbf{time BO} & \textbf{me.\ BO} & \textbf{time GW} & \textbf{me.\ GW} \\
    \hline
    \hline
%vasy\_25       & 25,217    & 25,216     & 25,217   & 25,216   & 12.97   & 65    & 0.11     & 55    & \textbf{0.16}    & 80    \\
vasy\_40       & 40,006    & 60,007     & 20,003   & 40,004   & 142.77  & 65    & 762.69   & 62    & \textbf{0.34}    & 93    \\
vasy\_65     & 65,537    & 2,621,480   & 65,536   & 2,621,440 & 239.67  & 437   & 47.88    & 645   & \textbf{20.07}   & 2,481  \\
vasy\_66     & 66,929    & 1,302,664   & 51,128   & 1,018,692 & \textbf{7.42}    & 208   & 16.16    & 356   & 9.05    & 853   \\
vasy\_69      & 69,754    & 520,633    & 69,753   & 520,632  & \textbf{3.98}    & 155   & 12.65    & 171   & 4.53    & 493   \\
%vasy\_83      & 83,436    & 325,584    & 42,195   & 197,200  & 12.62   & 113   & 3.30     & 119   & \textbf{2.11}    & 272   \\
vasy\_116     & 116,456   & 368,569    & 22,398   & 87,674   & 3.84    & 95    & 15.73    & 128   & \textbf{2.68}    & 142   \\
vasy\_157     & 157,604   & 297,000    & 3,038    & 12,095   & 6.98    & 97    & 6.80     & 110   & \textbf{1.08}    & 129   \\
vasy\_164    & 164,865   & 1,619,204   & 992     & 3,456    & \textbf{3.89}    & 251   & 20.20    & 316   & 5.38    & 246   \\
vasy\_166     & 166,464   & 651,168    & 42,195   & 197,200  & 21.60   & 153   & 6.20     & 177   & \textbf{3.89}    & 376   \\
cwi\_214      & 214,202   & 684,419    & 478     & 1,612    & \textbf{0.87}    & 140   & 29.92    & 197   & 2.64    & 140   \\
cwi\_371      & 371,804   & 641,565    & 2,134    & 5,634    & 42.70   & 179   & 17.37    & 261   & \textbf{3.12}    & 168   \\
cwi\_566     & 566,640   & 3,984,157   & 198     & 791     & 1683.28 & 454   & 26.24    & 531   & \textbf{19.94}   & 454   \\
vasy\_574   & 574,057   & 13,561,040  & 3,577    & 16,168   & 105.10  & 1,766  & 487.01   & 2,192  & \textbf{40.18}   & 1,495  \\
%vasy\_1112   & 1,112,490  & 5,290,860   & 265     & 1,300    & \textbf{19.49}   & 779   & 48.39    & 1,016  & 22.91   & 902   \\
cwi\_2165    & 2,165,446  & 8,723,465   & 4,256    & 20,880   & 80.56   & 1,403  & 387.93   & 2,409  & \textbf{59.49}   & 1,948  \\
cwi\_2416   & 2,416,632  & 17,605,592  & 730     & 2,899    & 1,679.55 & 1,932  & \textbf{59.29}    & 2,660  & 90.69   & 1,932  \\
vasy\_2581  & 2,581,374  & 11,442,382  & 704,737  & 3,972,600 & 2,592.74 & 1,690  & 463.52   & 2,344  & \textbf{76.16}   & 5,098  \\
vasy\_4220  & 4,220,790  & 13,944,372  & 1,186,266 & 6,863,329 & 3,643.08 & 2,054  & 863.74   & 2,951  & \textbf{119.20}  & 7,287  \\
vasy\_4338 & 4,338,672 & 15,666,588 & 704,737 & 3,972,600 & 5,290.54 & 2,258 & 587.87 & 3,026 & \textbf{109.21} & 6,927 \\
vasy\_6020  & 6,020,550  & 19,353,474  & 256     & 510     & 130.76  & 2,045  & 95.76    & 3,482  & \textbf{45.54}   & 2,045  \\
vasy\_6120  & 6,120,718  & 11,031,292  & 2,505    & 5,358    & 546.11  & 1,893  & 291.30   & 2,300  & \textbf{81.05}   & 3,392  \\
cwi\_7838   & 7,838,608  & 59,101,007  & 62,031   & 470,230  & 745.33  & 6,319  & 11,667.98 & 11,027 & \textbf{617.46}  & 14,456 \\
vasy\_8082  & 8,082,905  & 42,933,110  & 290     & 680     & 288.45  & 6,098  & 677.28   & 7,824  & \textbf{200.72}  & 6,108  \\
vasy\_11026 & 11,026,932 & 24,660,513  & 775,618  & 2,454,834 & 5,005.61 & 3,642  & 2,555.30  & 5,235  & \textbf{225.20}  & 10,394 \\
vasy\_12323 & 12,323,703 & 27,667,803  & 876,944  & 2,780,022 & 5,997.26 & 4,068  & 2,068.52  & 5,770  & \textbf{256.70}  & 11,575 \\
cwi\_33949 & 33,949,609 & 165,318,222 & 12,463   & 71,466   & 1,684.56 & 21,951 & 11,635.09 & 42,162 & \textbf{1,459.92} & 37,437 \\
dining\_14 & 18,378,370 & 164,329,284 & 228,486 & 2,067,856 & 1,264.67 & 20,155 & 3,010.17 & 31,201 & \textbf{1,100.91} & 20,155 \\
1394-fin3 & 126,713,623 & 276,426,688 & 160,258 & 538,936 & 229,217.0 & 26,000 & 15,319.00 & 75,000 & \textbf{1,516.00} & 45,000 \\
\hline
  \end{tabular}
\caption{Runtime (in sec.) and memory use (in MB) results for GV, BO, and GW}

\label{tab:results}

\end{table}
\end{landscape}

\bibliographystyle{plain}

\end{document}